\documentclass{llncs}
\usepackage{makeidx}  

\usepackage[linesnumbered,ruled,vlined]{algorithm2e}

\SetAlFnt{\small}
\SetAlCapFnt{\small}
\SetAlCapNameFnt{\small}
\SetAlCapHSkip{0pt}
\IncMargin{-\parindent}

\SetKwRepeat{Do}{do}{until}%

\usepackage[english]{babel}
\usepackage[utf8x]{inputenc}
\usepackage{amsmath}
\usepackage{enumerate}
\usepackage{url}
\usepackage{xparse}
\usepackage{epigraph}

\usepackage{thmtools}
\usepackage{thm-restate}

\let\originalleft\left
\let\originalright\right
\renewcommand{\left}{\mathopen{}\mathclose\bgroup\originalleft}
\renewcommand{\right}{\aftergroup\egroup\originalright}

\newcommand{\ceil}[1] {\left\lceil{#1}\right\rceil}
\newcommand{\floor}[1] {\left\lfloor{#1}\right\rfloor}
\newcommand{\size}[1] {\left\lvert{#1}\right\rvert}
\newcommand{\set}[1]{\{#1\}}
\newcommand{\half}[1]{\frac{#1}{2}}
\newcommand{\prob}[1]{\mbox{\rm Prob}\left[\, #1\,\right]}
\DeclareMathOperator*{\argmax}{\arg\!\max}

\DeclareDocumentCommand{\Exp}{ m g }{
    {%
    \mathrm{E}
        \IfNoValueF {#2} {_{#2}}
    \left[#1\right]
    }%
}
\newcommand{\theoremC}{3}
\newcommand{\gammaUP}{\gamma^*}
\newcommand{\gammaUPVal}{0.003176}

\newcommand{\theoremCC}{E_1}
\newcommand{\theoremCCVal}{7.4}
\newcommand{\ETwoVal}{5}
\newcommand{\ErrorConstSum}{12.4}

\begin{document}

\title{The Temp Secretary Problem}

\author{
    Amos Fiat\inst{1}, Ilia Gorelik\inst{1}, Haim Kaplan\inst{1}, Slava Novgorodov\inst{1}
}
\institute{Tel Aviv University, Israel,\\
\email{fiat@tau.ac.il, iliagore@post.tau.ac.il, haimk@tau.ac.il, slavanov@post.tau.ac.il}\\
Research supported by The Israeli Centers of Research Excellence (I-CORE) program (Center No. 4/11), and by ISF grant no.\ 822/10.
}


\maketitle

\begin{abstract}
We consider a generalization of the secretary problem where
contracts are temporary, and for a fixed duration $\gamma$. This models online hiring of temporary  employees, or online auctions for re-usable resources. The problem is related to the question of finding a large independent set in a random unit interval graph.
\end{abstract}


\section{Introduction}
This paper deals with a variant of the secretary model, where contracts are temporary.
{\sl E.g.}, employees are hired for short-term contracts, or re-usable resources are rented out repeatedly, etc. If an item is chosen, it ``exists" for a fixed length of time and then disappears.

Motivation for this problem are web sites such as Airbnb and oDesk. Airbnb offers short term rentals in competition with classic hotels. A homeowner posts a rental price and customers either accept it or not. oDesk is a venture capitalizing on freelance employees. A firm seeking short term freelance employees offers a salary and performs interviews of such employees before choosing one of them.

We consider an online setting where items have values determined by an adversary, (``no information" as in the standard model \cite{Freeman:1983}), combined with stochastic arrival times that come from a prior known distribution (in contrast to the random permutation assumption and as done in \cite{Karlin:1962,bruss1984,Gallego:1994}). Unlike much of the previous work on online auctions with stochastic arrival/departure timing (\cite{Hajiaghayi:2004}), we do not consider the issue of incentive compatibility with respect to timing, and assume that arrival time cannot be misrepresented.


The temp secretary problem can be  viewed
\begin{enumerate} \item As a problem related to hiring temporary workers of varying quality subject to workplace capacity constraints. There is some known prior $F(x)=\int_0^x f(z) dz$ on the arrival times of job seekers, some maximal capacity, $d$, on the number of such workers that can be employed simultaneously, and a bound $k$ on the total number than can be hired over time. If hired, workers cannot be fired before their contract is up.
\item Alternately, one can view the temp secretary problem as dealing with social welfare maximization in the context of rentals. Customers arrive according to some distribution. A firm with capacity $d$ can rent out up to $d$ boats simultaneously, possibly constrained to no more than $k$ rentals overall. The firm publishes a rental price, which may change over time {\sl after} a customer is serviced.  A customer will choose to rent if her value for the service is at least the current posted price. Such a mechanism is inherently dominant strategy truthful, with the caveat that we make the common assumption that customers reveal their true values in any case.
\end{enumerate}


We give two algorithms, both of which are quite simple and offer
posted prices for rental that vary over time. Assuming that the time
of arrival cannot be manipulated, this means that our
algorithms are dominant strategy incentive compatible.


For rental duration $\gamma$, capacity $d=1$, no budget restrictions, and arrival times from an arbitrary prior, the {\sl time-slice algorithm} gives a
$\frac{1}{2e}$ competitive ratio. For arbitrary  $d$ the competitive ratio of the time-slice algorithm is at least $(1/2) \cdot (1-5/\sqrt{d})$. This can be generalized to more complex settings, see Table \ref{tab:results3}.
The time slice algorithm  divides time into slices of length
$\gamma$. It randomly decides if to work on even or odd slices.
Within each slice it uses a variant of some other secretary problem
({\sl E.g.}, \cite{Lachish14}, \cite{Babaioff07-knapsack}, \cite{Kleinberg}) except that it keeps track of the
cumulative distribution function rather than the number of
secretaries.

The more technically challenging {\sl Charter algorithm} is strongly motivated by the $k$-secretary algorithm of \cite{Kleinberg}. For capacity $d$,  employment period $\gamma$, and budget $d\leq k\leq d/\gamma$ (the only relevant values), the Charter algorithm does the following:
\begin{itemize}
\item Recursively run the algorithm with parameters $\gamma, \lfloor k/2 \rfloor$  on all bids that arrive during the period $[0,1/2)$.
\item
Take the bid of rank $\ceil{k/2}$ that appeared during the period
$[0,1/2)$, if such rank exists and set a threshold $T$ to be it's
value. If no such rank exists set the threshold $T$ to be zero.
    \item
    Greedily accept all items that appear during the period $[1/2,1)$
that have value at least $T$ ---
  subject to not exceeding capacity ($d$) or budget ($k$) constraints.
  \end{itemize}
  For $d=1$ the competitive ratio of the Charter algorithm is at least $$\frac{1}{1+k\gamma}\left(1-\frac{5}{\sqrt{k}}-\theoremCCVal\sqrt{\gamma \ln(1/\gamma)}\right).$$ Two special cases of interest are $k=1/\gamma$ (no budget restriction), in which case the expression above is at least $\frac{1}{2}\left( 1-\ErrorConstSum \sqrt{\gamma \ln(1/\gamma)}\right)$. We also show an upper bound of $1/2 +\gamma/2$ for $\gamma>0$. As $\gamma$ approaches zero the two bounds converge to $1/2$. Another case of interest is when $k$ is fixed and $\gamma$ approaches zero in which this becomes the guarantee given by Kleinberg's $k$-secretary algorithm.

  For arbitrary $d$ the competitive ratio of the Charter algorithm is at least $$1-\Theta\left(\frac{\sqrt{\ln d}}{\sqrt{d}}\right) - \Theta\left(\gamma\log{(1/\gamma)}\right).$$

  We remark that neither the time slice algorithm nor the Charter algorithm requires prior knowledge of $n$, the number of items due to arrive.

At the core of the analysis of the Charter algorithm we prove a
bound on the expected size of the maximum independent set of a
random unit interval graph. (See Table \ref{tab:results1}). In this random graph model we draw $n$
intervals, each of length $\gamma$, by drawing their left endpoints
uniformly in the interval $[0,1)$. We prove that the expected size
of a maximum independent set in such a graph is about
$n/(1+n\gamma)$. We say that a set of length $\gamma$ segments that do not overlap is $\gamma$-independent. Similarly, a capacity $d$ $\gamma$-independent set allows no more than $d$ segments overlapping at any point.

Note that if $\gamma = 1/n$ then this expected
size is about $1/2$. This is intuitively the right bound as  each
interval in the maximum independent set rules out on average one
other interval from being in the maximum independent set.

We show that a random unit interval graph with $n$ vertices has a capacity $d$ $\gamma$-independent subset of expected size at least $\min(n,d/\gamma)(1-\Theta({\sqrt{\ln{d}}}/{\sqrt{d}}))$. We also show that when $n=d/\gamma$ the expected size of the maximum capacity $d$ $\gamma$-independent subset is no more than $n(1-\Theta({1}/{\sqrt{d}}))$. These
results may be of independent interest.



\smallskip

\noindent
{\bf Related work:}
Worst case competitive analysis of interval scheduling has a long history, e.g., \cite{Woeginger19945,Lipton:1994:OIS}. This is the problem of choosing a set of non-overlapping intervals with various target functions, typically, the sum of values.

\cite{Hajiaghayi:reuse05} introduce the question of auctions for reusable goods. They consider a worst case mechanism design setting. Their main goal is addressing the issue of time incentive compatibility, for some restricted set of misrepresentations.
%

The secretary problem is arguably due to Johannes Kepler (1571-1630),
and has a great many variants, a survey by \cite{Freeman:1983} contains some 70 references. The ``permutation" model is that items arrive in some random order, all $n!$ permutations equally likely.
Maximizing the probability that the best item is chosen, when the items appear in random order, only comparisons can be made, and the number of items is known in advance, was solved by \cite{Lindley:1961} and by \cite{Dynkin:1963}.
A great many other variants are described in (\cite{Freeman:1983,Dinitz:2013}), differing in the number of items to be chosen, the target function to be maximized, taking discounting into account, etc.


An alternative to the random permutation model is the stochastic arrival model, introduced by Karlin \cite{Karlin:1962}
 in a ``full information" (known distribution on values) setting.
Bruss \cite{bruss1984} subsequently studied the stochastic arrival model in a no-information model (nothing is known about the distribution of values). Recently, \cite{FNS11} made use of the stochastic arrival model as a tool for the analysis of algorithms in the permutation model.

Much of the recent interest in the secretary problem is due to it's connection to incentive compatible auctions and posted prices \cite{Hajiaghayi:2004,Kleinberg,Babaioff07-knapsack,Babaioff08,Babaioff:2009:SPW,devanur2009adwords}.

Most directly relevant to this paper is the $k$-secretary algorithm by R. Kleinberg
\cite{Kleinberg}. Constrained to picking no more than $k$ secretaries, the total value of the secretaries picked by this algorithm is at least a $(1-\frac{5}{\sqrt{k}})$ of the value of the best  $k$ secretaries.

 Babaioff {\sl et al.} \cite{Babaioff07-knapsack} introduced the {\em knapsack secretary problem} in which every secretary has some weight and a
value, and one seeks to maximize the sum of values subject to a
upper bound on the total weight. They give a $1/(10e)$ competitive algorithm for this problem. (Note that if weights are one then this becomes the $k$-secretary problem).
 The Matroid
secretary problem, introduced by Babaioff et al. \cite{DBLP:conf/soda/BabaioffIK07}, constrains the set of secretaries picked to be
an independent set in some underlying Matroid.
Subsequent results for arbitrary Matriods are given in \cite{Chakraborty:2012,Lachish14,FeldmanSZ15}.



Another generalization of the secretary problem is the online
maximum bipartite matching problem. See \cite{KorulaP09,KesselheimRTV13}.
Secretary models with full information or partial information (priors on values) appear in \cite{Badanidiyuru:2012:LBP} and \cite{Singer:2013:PMC}. This was in the context of submodular procurement auctions (\cite{Badanidiyuru:2012:LBP}) and budget feasible procurement (\cite{Singer:2013:PMC}).
Other papers considering a stochastic setting include \cite{KesselheimTRV14,HKSV14}.

In our analysis, we give a detailed and quite technical lower bound on the
size of the maximum independent set in a random unit interval graph
(produced by the greedy algorithm). Independent sets in other random interval graph models were previously studied in \cite{winkler1990,winkler95,CPC:91295}.
%

\section{Formal Statement of Problems Considered} \label{sec:def}
\newcommand{\opt}{\mbox{\rm Opt}}

Each item $x$ has a value $v(x)$, we assume that for all $x\neq y$, $v(x)\neq v(y)$ by consistent tie breaking, and we say that $x>y$ iff $v(x)>v(y)$. Given a set of items $X$, define $v(X)=\sum_{x\in X} v(x)$ and $T_k(X) = \max_{T\subseteq X, |T|\leq k} v(T)$.

Given a set $X$ and a density distribution function $f$ defined on $[0,1)$, let $\theta_f: X \mapsto [0,1)$
be a random mapping where $\theta_f(x)$ is drawn independently from the distribution $f$. The function $\theta_f$ is called a {\sl stochastic arrival} function, and we interpret $\theta_f(x)$, $x\in X$, to be the time at which item $x$ arrives.
For the special case in which $f$ is uniform we refer to $\theta_f$ as $\theta$.

In the problems we consider, the items arrive in increasing order of $\theta_f$. If $\theta_f(x)=\theta_f(y)$ the relative
order of arrival of $x$ and $y$ is arbitrary.
An online algorithm may select an item only upon arrival. If an item $x$ was selected, we say that the online algorithm
{\em holds} $x$ for $\gamma$ time following $\theta_f(x)$.

 An online algorithm $A$ for the temp secretary problem may
hold at most one item at any time and may select at most $k$ items in total. We refer to  $k$ as the {\em budget} of $A$.
 The goal of the algorithm is to maximize the expected total value of the items that it selects. We denote by $A(X,\theta_f)$  the set
 of items chosen by algorithm $A$ on items in $X$ appearing according to stochastic arrival function $\theta_f$.

The set of the arrival times of the items selected by an algorithm for the temp secretary problem is said to be \emph{$\gamma $-independent}.
Formally, a set $S \subset [0,1)$ is said to be \emph{$\gamma $-independent}
if for all $t_1,t_2\in S$, $t_1\not=t_2$ we have that $|t_1-t_2|\geq \gamma$.

Given $\gamma>0$, a budget $k$, a set $X$ of items, and a mapping $\theta_f : X\mapsto [0,1)$ we define $\opt(X, \theta_f)$
to be a $\gamma$-independent set $S$, $|S|\leq k$, that maximizes the sum of values.

Given rental period $\gamma>0$, distribution $f$, and budget $k$, the competitive ratio of an online algorithm $A$ is defined to be
\begin{equation}\label{def:comp_ration}\inf_X{\frac{\Exp{v(A(X, \theta_f))}{\theta_f : X\mapsto [0,1)}}{\Exp{v(\opt(X, \theta_f))}{\theta_f : X\mapsto [0,1)}}}.\end{equation}
The competitive ratio of the temp secretary problem is the supremum over all algorithms $A$ of the competitive ratio of $A$.

Note that when $\gamma \rightarrow 0$, the
the temp secretary problem reduces to Kleinberg's $k$-secretary problem.

We extend the $\gamma$-temp secretary problem by allowing the algorithm to hold at most $d$ items at any time.
Another extension we consider is the {\em knapsack temp secretary problem} where each  item has a weight and
we require the set held by the algorithm at any time to be of total weight at most $W$.
Also, we define the {\em Matroid  temp secretary problem} where one restricts the set of items held by the algorithm at any time
to be an independent set in some Matroid $M$.

More generally, one can define a temp secretary problem with respect to some arbitrary predicate $P$ that holds on the set of items held by an online algorithm at all times $t$.
This framework includes all of the variants above.  The optimal solution with respect to $P$ is also well defined.

\section{The time-slice Algorithm.}\label{sec:simple}

In this section we describe a simple time slicing technique.
This gives a reduction from  temp secretary problems, with arbitrary known prior distribution on arrival times, to the ``usual" continuous setting where
secretaries arrive over time, do not depart if hired, and the distribution on arrival times is uniform. The reduction is valid for many variants of the temp secretary problem, including the Matroid secretary problem, and the knapsack secretary
problem. We remark that although the Matriod and Knapsack algorithms are stated in the random permutation model, they can be replaced with analogous  algorithms in the continuous time model and can therefore be used in our context.

We demonstrate this technique by applying it to the classical secretary problem (hire the best secretary). We
obtain an algorithm which we call
$Slice_\gamma$ for the temp secretary problem with arbitrary prior distribution on arrival times that is
 $O(1)$
competitive.


Consider the $1/2\gamma$ time intervals ({\sl i.e.} slices) $I_j=[2\gamma j ,~ 2\gamma (j+1))$,  $0 \leq j \leq 1/(2\gamma)-1$. We  split every such
interval into two, $I_j^\ell=[2\gamma j ,~ 2\gamma j + \gamma)$,
$I_j^r = [2 \gamma j +\gamma,~ 2 \gamma (j+1))$.\footnote{For simplicity we assume that $1/(2\gamma)$ is an integer.}

Initially, we flip a fair coin and with probability $1/2$ decide to
pick points only from the left halves ($I_j^\ell$'s)
or only from the right halves ($I_j^r$'s). In each such interval we pick at most
one item by running the following modification of the continuous time secretary
algorithm.

The continuous time secretary algorithm \cite{FNS11} observes the items arriving before time $1/e$, sets the largest value of an observed item as
 a threshold, and then chooses
the first item (that arrives following time $1/e$) of value greater than the threshold.
 The modified continuous time secretary algorithm  observes items  as long as the cumulative distribution function of the current time is less than $1/e$, then it sets the largest value of an observed item as
 a threshold
compute a threshold, and  choose the next
item of value larger than the threshold.

It is clear that any two points picked by this algorithm have
arrival times separated by at least $\gamma$.

\begin{theorem}\label{THEOREM:SIMPLE}
The algorithm $Slice_\gamma$ is $1/(2e)$ competitive.
\end{theorem}
\begin{proof}
The analysis is as follows. Fix the mapping of items to each of the
left intervals $I_j^\ell$'s and to each of the right intervals
$I_j^r$'s (leaving free the assignment of items to specific arrival
times within their the intervals they are assigned to). Let
$OPT^\ell$ ($OPT^r$) be the sum of the items of maximum value over
all intervals $I_j^\ell$ ($I_j^r$). Let $OPT$ be the average optimal
value conditioned on this mapping of items to intervals. Clearly,
\begin{equation}\label{eq:sumOps}
OPT^\ell + OPT^r \geq OPT.
\end{equation}

For any interval $I_j$'s ($I_j^\ell$'s) $Slice_\gamma$ gain at least $1/e$ over the top value in the interval
conditioned on the event that $Slice_\gamma$ doesn't ignore this interval, this happens with probability $1/2$.
Therefore the expected sum of values achieved by $Slice_\gamma$ is at least
\begin{equation}\label{eq:slicePreLast}
\frac{1}{2} \cdot \frac{1}{e}OPT^\ell  + \frac{1}{2}\cdot \frac{1}{e} OPT^r\ .
\end{equation}
Substitution (\ref{eq:sumOps}) into (\ref{eq:slicePreLast}) we get the lemma.
\qed\end{proof}

Appropriately choosing times (rather than number of elements) as a
function of the prior distribution allows us to do the same for
other variants of the secretary problem, the Knapsack (achieving a competitive ratio of $\frac{1}{2}\cdot \frac{1}{10e}$, see \cite{Babaioff07-knapsack}) and Matriod ($O(\ln\ln\rho)$ when $\rho$ is the rank of the Matroid, see \cite{Lachish14,FeldmanSZ15}).

\section{Improved results for the temp secretary problem for the uniform arrival distribution}\label{sec:charterUniform}

In this section we give an improved algorithm, referred as the charter algorithm $C_{k,\gamma}$, for the temp secretary problem with uniform arrival times and capacity $1$ (at most one secretary can be hired at any time).

As it is never the case that more than $1/\gamma$ items can be selected, setting $k=\ceil{1/\gamma}$ effectively removes the budget constraint.
Note that $C_{k,0}$ is Kleinberg's algorithm for the $k$-secretary
problem, with some missing details added to the description.

To analyze the charter algorithm we establish
a lower bound on the
expected size of the maximum $\gamma$-independent subset of a set of
uniformly random points in $[0,1)$.
We apply this lower bound to  the subset of
the items that Kleinberg's algorithm selects.

\subsection{The temp secretary algorithm, $C_{k,\gamma}$: a competitive ratio of $~1/(1+k\gamma)$.}
\label{sec:charter}

This charter algorithm, $C_{k,\gamma}$ gets parameters $k$ (the maximal
number of rentals allowed) and $\gamma$ (the rental period) as is described in detail in Algorithm \ref{algo:tempalg}. As the
entire period is normalized to $[0,1)$, having  $k>\ceil{1/\gamma}$ is irrelevant. Thus, we assume that $k\leq
\ceil{1/\gamma}$.\footnote{To simplify the presentation we shall assume the in sequel that $k\leq
1/\gamma$.}

We show that $C_{k,\gamma}(X)$ gains in expectation about $1/(1+k\gamma)$ of the
top $k$ values of $X$, which implies that the competitive ratio (see definition (\ref{def:comp_ration})) of
$C_{k, \gamma}$ is at least about $1/(1+k\gamma)$.

Note that for $k=\ceil{1/\gamma}$,  $C_{k, \gamma}$ has a competitive ratio close to $1/2$,
while for $\gamma = 0$, $C_{k, \gamma}$ has a competitive ratio close to $1$.

%
\begin{algorithm}[h]\label{algo:tempalg}
    \caption{The Charter Algorithm $C_{k,\gamma}$.}
    \DontPrintSemicolon
        \uIf{$k=1$}{
            \tcc{Use the ``continuous secretary'' algorithm \cite{FNS11}:}
            Let $x$ be the largest item to arrive by time $1/e$ (if no item arrives by time $1/e$ --- let $x$ be the absolute zero, an item smaller than all other items).\;
            $C_{k,\gamma}$ accepts the first item $y$, $y>x$, that arrives after time $1/e$ (if any) \;
        }
        \Else{
            \tcc{Process the items scheduled during the time interval $[0,\,1/2)$}
            Initiate a recursive copy of the algorithm, $C'=C_{\floor{k/2},2\gamma}$.\;
            $x \gets $ next element \tcp*{If no further items arrive $x\gets \emptyset$}
            \While{$x \neq  \emptyset$ ~AND~ $\theta(x) < 1/2$} {
                Simulate $C'$ with input $x$ and modified schedule $\theta'(x) = 2 \theta(x)$.\;
                \If{$C'$ accepts $x$}{
                    $C_{k,\gamma}$ accepts $x$\;
                }
                $x \gets $ next element \tcp*{If no further items arrive, $x\gets \emptyset$}
            }
            \tcc{Determine threshold $T$ }
            Sort the items that arrived during the time interval $[0,\,1/2)$: $y_1> y_2> \cdots> y_{m}$ (with consistent tie breaking). \;
            Let $\tau = \ceil{k/2}$.\;
            \If{$m < \tau$} {
                set $T$ to be the absolute zero
            }
            \Else {
                set $T\gets y_{\tau}.$ \;
            }
            \tcc{Process the items scheduled during the time interval $[1/2,\,1)$}
            \Do{$x = \emptyset$ ~OR~ $k$ items have already been accepted}{

                \If{$x > T$ ~AND~ ($\theta(x)\geq \theta(x')+\gamma$ where $x'$ is the last item accepted by $C_{k,\gamma}$
                    \\\qquad\qquad ~OR~ no items have been previously accepted)}{
                   $C_{k,\gamma}$ accepts $x$\;
                }
                $x \gets $ next element \tcp*{If no further items arrive, $x\gets \emptyset$}
            }
        }
\end{algorithm}

It is easy to see that  $C_{k,\gamma}$, chooses a $\gamma$-independent set of size at most $k$.

The main theorem of this paper is the following generalization of Kleinberg's $k$-secretary problem:

\begin{restatable}{theorem}{mainThm}
\label{THEOREM:MAIN}
For any set of items $S = \set{x_i}_{i=1}^n$, $0 < \gamma \le \gammaUP = \gammaUPVal$ and any positive integer $k \leq 1/\gamma$:
\begin{equation}\label{eq:mainTheorem}
 \Exp{v(C_{k,\gamma}(S, \theta))}{\theta:S \mapsto [0,1]}
\geq \frac{1}{1+\gamma k} \left(1 - \beta(\gamma,k)\right)T_k(S),
\end{equation}
where
$\beta(\gamma,\textbf{\textbf{\textbf{}}}k) = \theoremCCVal\sqrt{\gamma\ln(1/\gamma)} + \frac{\ETwoVal}{\sqrt{k}},$
and the expectation is taken oven all uniform mappings of $S$ to the interval $[0,1)$.
(Note that  the right hand side of Equation (\ref{eq:mainTheorem}) is negative for $\gammaUP < \gamma \leq 0.5$.)
\end{restatable}

\subsection{Outline of the proof of Theorem~\ref{THEOREM:MAIN}}\label{sec:outline}

We prove Theorem \ref{THEOREM:MAIN} by induction on $k$. For $k \le
25$ the theorem holds vacuously.

The profit, $p^{[0,1/2)}$, on those items that arrive during the
time interval $[0,1/2)$ is given by the inductive hypothesis\footnote{This profit, $p^{[0,1/2)}$ is $\Exp{v(C^{[0,1/2)}_{k,\gamma}(S, \theta))}{\theta:S \mapsto [0,1)}$, where $C^{[0,1/2)}_{k,\gamma}(S, \theta)$ the set of items chosen by the algorithm during the time period $[0,1/2)$.}.
However, the inductive hypothesis gives this profit, $p^{[0,1/2)}$,
in terms of the top $\lfloor k/2 \rfloor$ elements that arrive
before time $1/2$, and not in terms of $T_k(X)$, the value of the top $k$ items overall. Thus, we need to relate $p^{[0,1/2)}$ to $T_k(X)$. In Lemma
\ref{lem:Bobby1}  we show that $p^{[0,1/2)}$ is about 1/2 of $T_k(X)$.

Let $Z_{>T}$ be the set of items that arrive in the time interval
$[1/2,1)$ and have value greater than the threshold $T$. From $Z_{>T}$ we
greedily pick a $\gamma$-independent subset\footnote{modulo the caveat that the arrival time of
the 1st item chosen from the 2nd half must be at least $\gamma$
after  the arrival time of the last item chosen in the 1st half.}. It is easy to see that this  set is in
fact  a maximal $\gamma$-independent subset.

To bound the expected profit  from the items in $Z_{>T}$
 we first bound the size of the maximal
$\gamma$-independent set amongst these items. To do so we
use the following general theorem (see also Section \ref{sec:independent} and Section \ref{sec:half_details}).

\begin{restatable}{theorem}{halfThm}
\label{THEOREM:HALF}
Let $Z=\set{z_1, z_2, \ldots, z_n}$ be a set of independently uniform
samples, $z_i$, from the real interval $[0,1)$. For $0 \leq \gamma
\leq 1$,
\begin{equation}
\Exp{m(Z,\gamma)}{Z} \geq \frac{1-\alpha(\gamma)}{\gamma+1/n}= \frac{(1-\alpha(\gamma)) n}{1+n \gamma}, \quad \mbox{\rm where\ } \alpha(\gamma) = \theoremC\sqrt{\gamma\ln(1/\gamma)},\label{eq:theorem}
\end{equation}
 where
$m(Z,\gamma)$ denotes the size of the largest
$\gamma$-independent subset of $Z$.
\end{restatable}

We apply Theorem \ref{THEOREM:HALF} to the items in $Z_{>T}$. We can apply this theorem since arrival times of items in $Z_{>T}$ are uniformly distributed in the 2nd
half. Specifically, we give a lower bound on the expected profit of the algorithm from the items in the 2nd half as follows:

\begin{enumerate}
\item
 Condition on the size of $Z_{>T}$.
\item
Subsequently, condition on the set of arrival times $\{ \theta_1, \theta_2, \ldots, \theta_{|Z_{>T}|}\}$ of the items in $Z_{>T}$  but
{\sl not} on which item in $Z_{>T}$ arrives when. This conditioning fixes the $\gamma$-independent set selected greedily by the algorithm.
\item
 We take the expectation over all bijections $\theta$ whose image on the domain $Z_{>T}$ is the set  $\{ \theta_1, \theta_2, \ldots, \theta_{|Z_{>T}|}\}$. The expected profit (over the set $Z_{>T}$ and over
these bijections) is ``approximately"
\begin{equation}\frac{\mbox{\rm Size of maximal $\gamma$-independent set from $Z_{>T}$}}{|Z_{>T}|} \cdot \sum_{z \in  Z_{>T}} v(z).\label{eq:expectedprofit}\end{equation}
The ``approximately" is because of some technical difficulties:
\begin{itemize}
  \item We cannot ignore the last item amongst those arriving prior to time 1/2. If one such item was chosen at some time  $1/2-\gamma< t < 1/2$ then arrivals during the period $[1/2,t+\gamma)$ cannot be chosen.
\item We cannot choose more than $k$ items in total, if the algorithm choose $\lambda$ items from the time interval $[0,1/2)$, it cannot choose more than $k-\lambda$ items from the time interval $[1/2,1)$, but $k-\lambda$ may be smaller than the size of the $\gamma$-independent set from $Z_{>T}$.
\end{itemize}
\item To get an unconditional lower bound we average Equation (\ref{eq:expectedprofit}) over the possible sizes of the $\gamma$-independent set as given by Theorem \ref{THEOREM:HALF}.
\end{enumerate}

\section{Upper bound for the temp secretary problem with uniform arrival times and with no budget restriction}
\label{sec:upperTempSec}
\begin{theorem}\label{THEOREM:UPPERBOUND}
For the temp secretary problem where item arrival times are taken from the uniform distribution, for any $\gamma \in (0,1)$, any online algorithm (potentially randomized) has a competitive ratio  $\leq 1/2+\gamma/2$.
\end{theorem}
\begin{proof}
Let $A$ denote the algorithm.
Consider the following two inputs:
\begin{enumerate}
    \item The set $S$ of n-1 items of value $1$.
    \item The set $S' = S\cup\set{x_n}$ where $v(x_n) = \infty$.
\end{enumerate}
Note that these inputs are not of the same size (which is ok as the number of items is unknown to the algorithm).

Condition the mapping $\theta:S\mapsto [0,1)$ (but not the mapping of $x_n$). If $A$ accepts an item $x$ at time $\theta(x)$ we say that the segment $[x,x+\gamma)$ is \emph{covered}.
For a fixed $\theta$ let $g(\theta)$ be the expected  fraction of $[0,1)$ which is not covered when running $A$ on the set $S$ with arrival times
  $\theta$.
This expectation is over the  coin tosses of $A$.
Let $G$ be $\Exp{g(\theta)}{\theta:S\mapsto[0,1)}$.

The number of items that $A$ picks on the input $S$ with arrival time $\theta$ is at most $\frac{1-g(\theta)}{\gamma} + 1$. Taking expectation over all mappings $\theta:S\mapsto [0,1)$ we get that the value gained by $A$ is at most $(1-G)/\gamma + 1$.

As $n\rightarrow \infty$ the optimal solution consists of $\ceil{1/\gamma}$ items of total value $\ceil{1/\gamma}$.
Therefore the competitive ratio of $A$ is at most
\begin{equation}\label{eq:bound1}
\frac{(1-G)/\gamma + 1}{1/\gamma} = 1-G+\gamma.
\end{equation}

Note that $g(\theta)$ is exactly the probability that $A$ picks $x_n$ on the input $S \cup \set{x_n}$
(this probability is over the mapping of $x_n$ to $[0,1)$ conditioned upon the arrival times of all the items in $S\subset S'$). Therefore the competitive ratio of $A$ on the
input $S'$ is
\begin{equation}\label{eq:bound2}
\Exp{g(\theta)} = G\ .
\end{equation}
Therefore the competitive ratio of $A$ is no more than the minimum of the two upper bounds (\ref{eq:bound1}) and (\ref{eq:bound2})
$$\min{(G,\, 1-G+\gamma)}\leq 1/2+\gamma/2 \ .$$
\qed\end{proof}

\section{About Theorem \ref{THEOREM:HALF}: A Lower bound on the expected size of the maximum ${\gamma}$-independent subset}\label{sec:independent}

\newcommand{\theoremCOne}{c_1}
\newcommand{\theoremCOneVal}{2.058664}
\newcommand{\theoremCTwo}{c_2}
\newcommand{\theoremCTwoVal}{0.620670}
\newcommand{\theoremCThree}{c_3}
\newcommand{\theoremCThreeVal}{0.293353}

\newcommand{\theoremTVal}{0.032704}
\newcommand{\theoremT}{t}

\newcommand{\theoremInvT}{1/t}
\newcommand{\theoremInvTVal}{30.577509}
\newcommand{\theoremInvTValApprox}{31}

\newcommand{\theoremB}{\beta(\gamma)}
\newcommand{\theoremBVal}{1.829048}

\newcommand{\theoremH}{h}
\newcommand{\theoremHVal}{0.463713}

\newcommand{\theoremK}{k}
\newcommand{\theoremKVal}{2.490795}

\newcommand{\theoremRGamma}{r(\gamma)}

\newcommand{\theoremF}{f}
\newcommand{\theoremFVal}{4}

\newcommand{\theoremFMinOneVal}{3}
\newcommand{\theoremTwoFMinOneVal}{7}
\newcommand{\theoremFOverF}{\frac{f}{f-1}}
\newcommand{\theoremFOverFVal}{\frac{4}{3}}
\newcommand{\theoremFOverOverFVal}{\frac{3}{4}}

Recall the definition of $Z$ and $m(Z,\gamma)$ from Theorem \ref{THEOREM:HALF}.

Define the random variable $X_i$, $1 \leq i \leq n$ to be the $i$'th smallest point in $Z$. Define the random variable
$C_i$ to be the number of points from $Z$ that lie in the interval
$[X_i,X_i+ \gamma)$. Note that at most one of these points can belong to a $\gamma$-independent set.

The \emph{greedy algorithm} constructs a maximal $\gamma$-independent set by
traversing points of $Z$ from the small to large and adding a point whenever possible.
Let $I_i$ be a random variable with binary values where $I_i=1$ iff $X_i$ was chosen by the greedy algorithm. It follows from the definition that $\sum_i I_i$ gives the size of the maximal independent set, $m(Z,\gamma)$, and that $\sum_i I_i C_i =n$.

Note that $\Exp{C_i}\leq 1+ n \gamma$, one for the point $X_i$ itself, and $n \gamma$ as the expected number of uniformly random points that fall into an interval of length $\gamma$.
If $C_i$ and $I_i$ were independent random variables, it would follow that
$$\Exp{\sum I_i C_i} \leq (1+n \gamma)\sum \prob{I_i=1},$$ and, thus,
$$m(Z,\gamma) = \sum I_1 \geq n/(1+n\gamma).$$

Unfortunately, $C_i$ and $I_i$ are not independent, and the rest of the proof of Theorem \ref{THEOREM:HALF} in Section \ref{sec:half_details} primarily deals with showing that this dependency is insignificant.


\section{Discussion and Open Problems}

We've introduced online optimization over temporal items under stochastic inputs subject to conditions of two different types:
\begin{itemize}
  \item ``Vertical" constraints: Predicates on the set of items held at all times $t$. In this class, we've considered conditions such as no more than $d$ simultaneous items held at any time, items held at any time of total weight $\leq W$, items held at any time must be independent in some Matroid.
  \item ``Horizontal" constraints: Predicates on the set of items over all times. Here, we've considered the condition that no more than $k$ employees be hired over time.
\end{itemize}
One could imagine much more complex settings where the problem is defined by arbitrary constraints of the first type above, and arbitrary constraints of the 2nd type. For example, consider using knapsack constraints in both dimensions. The knapsack constraint for any time $t$ can be viewed as the daily budget for salaries. The knapsack constraint over all times can be viewed as the total budget for salaries. Many other natural constraints suggest themselves.

It seems plausible that the time slice algorithm can be improved, at least in some cases, by making use of information revealed over time, as done by the Charter algorithm.


\bibliographystyle{plain}



\newpage

\appendix
\section*{APPENDIX}

\section{Table of results}

\begin{table}
\caption{Size of maximal $\gamma$-independent set. Uniformly prior on arrivals. } \label{tab:title} \label{tab:results1}
  {
    \begin{tabular}{|l|l|l|}
  \hline & &\\
     \begin{tabular}{l}
       Size of maximal \\
        $\gamma$-independent set\\ capacity $d=1$
           \end{tabular}
      & $\frac{n}{1+\gamma\cdot n} \left(1-\Theta\left(\sqrt{\gamma \ln(1/\gamma)}\right)\right)$ & Theorem \ref{THEOREM:HALF} \\ & &\\
      \hline & &\\
      \begin{tabular}{l}
       Size of maximal  \\
        $\gamma$-independent set  \\
        of capacity $d$ \\
         $n=d \cdot 1/\gamma$.
           \end{tabular}
           &
            \begin{tabular}{l}
                $\geq n \left(1-\Theta\left(\frac{\sqrt{\ln d}}{\sqrt{d}}\right)\right)$ \vspace*{.20in} \\
                $ \leq n \left(1-\Theta\left(\frac{1}{\sqrt{d}}\right)\right)$
           \end{tabular}
           & \begin{tabular}{l}
                Lower bound --- Theorem \ref{theorem:DWidthIndependetSet} \\
                Lower bound --- Theorem \ref{theorem:DWidthIndependetSetUpper}
           \end{tabular}\\ & &\\
     \hline  & & \\
      \begin{tabular}{l}
       Size of maximal  \\
        $\gamma$-independent set \\ of capacity $d$
           \end{tabular} & \begin{tabular}{l} $\geq \min(n,d/\gamma) \left(1-\Theta\left(\frac{\sqrt{\ln d}}{\sqrt{d}}\right)\right)$ \\
            $\leq \min(n,d/\gamma)$ \end{tabular}& Theorem \ref{theorem:DWidthIndependetSet} \\ & &\\
     \hline
   \end{tabular}
   }
\end{table}

\begin{table}
\caption{Competitive ratios, arbitrary prior on arrivals} \label{tab:results3}
  {
      \begin{tabular}{|l|l|l|}
      \hline & &\\
        \begin{tabular}{l}
            Capacity one $\gamma$-independent subset
        \end{tabular}
          & $\frac{1}{2e}$
          & Theorem \ref{THEOREM:SIMPLE} \\
          & &\\
          \hline
          & &\\
          \begin{tabular}{l}
         Matroid constraints
        \end{tabular}
          & $\Omega(1/\ln \ln \rho)$
          & \cite{Lachish14} \& Theorem \ref{THEOREM:SIMPLE} \\
          & & \\
              \hline & &\\
          \begin{tabular}{l}
         Knapsack constraints
                \end{tabular} & $\frac{1}{20e}$
                & \cite{Babaioff07-knapsack} \&  Theorem \ref{THEOREM:SIMPLE}\\
                & &\\
                \hline & & \\
                \begin{tabular}{l}
                    Capacity $d$ \\
                    $\gamma$-independent set
               \end{tabular}
               & $\frac{1}{2}\left(1 - \frac{5}{\sqrt{d}}\right)$
               &
                    \cite{Kleinberg} \&
                     Theorem \ref{THEOREM:SIMPLE}
               \\ & &\\
         \hline
       \end{tabular}
  }
\end{table}

\section{Proving the main Theorem}\label{sec:technical}
\mainThm*

\begin{proof}

The proof is via induction on $k$.
 For the base case of the induction we note that for any $ k < 25$ the theorem holds since $\left(1- \frac{\ETwoVal}{\sqrt{k}}\right) \leq 0$.
We hereby assume that  the theorem  holds for any $k' \leq k$, and prove the statement for $k$.

The proof is presented top-down. We refer to Lemmata
\ref{lem:sizeDoesntMatter}, \ref{lem:Bobby1} and \ref{lem:second-half},
whose statement and proof appear subsequently.

Let $C^{[0, 1/2)}\left(X, \theta\right)$ and $C^{[1/2, 1)}\left(X, \theta\right)$ be the subsets of $X$ chosen by algorithm $C_{k,\gamma}$
during the time intervals $[0, 1/2)$ and $[1/2, 1))$ respectively when applied to $X,\theta$.

For the induction step we use the  fact that
\begin{eqnarray}
  \lefteqn{\Exp{v\left(C_{k,\gamma}(S,\theta)\right)}{\theta:S\mapsto [0,1)}} \nonumber \\
  &=
    \Exp{v\left(C_{k,\gamma}^{[0,1/2)}(S,\theta)\right)}{\theta:S\mapsto [0,1)}
    +\Exp{v\left(C_{k,\gamma}^{[1/2,1)}(S,\theta)\right)}{\theta:S\mapsto [0,1)} \label{eq:ind11}
\end{eqnarray}
We give a lower bound on  the first term in (\ref{eq:ind11}) using the induction hypothesis (after an appropriate transformation)
and we directly lower bound the second term using Lemma \ref{lem:second-half}.

Since Lemma \ref{lem:second-half} requires that the size $n$ of $S$ is sufficiently large relative to $k$ (which may not be the case)
 we introduce a modification of $C_{k,\gamma}$, Algorithm~\ref{algo:zeroesC} --- denoted by $C_{k,\gamma}^*$:
\begin{algorithm}\label{algo:zeroesC}
    \DontPrintSemicolon

        During the time interval $[0,1/2)$, $C_{k,\gamma}^*$ emulates $C_{k,\gamma}$, {\sl i.e.}, $C_{k,\gamma}^{*[0,1/2)}(S) = C_{k,\gamma}^{[0,1/2)}(S)$.\;
        $D\gets$ a collection of $3\ceil{k/2}$   distinguishable dummy items with zero value.\;
        $\theta_D \gets $ random uniform mapping $D\mapsto [0,1)$\;
        $U \gets \set{x\in S \mid \theta(x)<1/2} \cup \set{d\in D \mid \theta_D(d)<1/2}.$\;
        \tcc{set the threshold $T^*$}
        $\tau = \ceil{k/2}$\;
        \uIf {$\size{U}\geq\tau$}{
             $T^*\gets $ the $\tau$ largest value amongst $U$\;
             \tcc{Note that if the number of items from $S$ mapped to the interval $[0,1/2)$ is $\geq \tau$, then $C_{k,\gamma}^*$ sets the same threshold as $C_{k,\gamma}$}
        }
        \Else{
            $T^*\gets $ absolute zero \tcp*{smaller than any other item}
        }
        \tcc{the arrival time of the next item will be after time $1/2$}
        $x \gets $ next item from $S\cup D$ \tcp*{If no further items arrive, $x\gets \emptyset$}
        \While{$x \neq \emptyset$ ~AND~ no more than $k-1$ items have already been accepted}{
            \If{$x > T^*$ ~AND~ ($\theta(x)\geq \theta(x')+\gamma$ where $x'$ is the last item accepted by $C_{k,\gamma}^{*}$
                    \\\qquad\qquad ~OR~ no items have been previously accepted)}{
                   $C_{k,\gamma}^{*}$ accepts $x$\;
                }
                $x \gets $ next item from $S\cup D$\tcp*{If no further items arrive, $x\gets \emptyset$}
        }
    \caption{$C_{k,\gamma}^*$}
\end{algorithm}


By Lemma \ref{lem:sizeDoesntMatter} we have that
\begin{equation}\label{eq:sizeValueCmp}
\Exp{v\left(C_{k,\gamma}(S,\theta)\right)}{\theta:S\mapsto [0,1)}
\geq \Exp{v\left(C_{k,\gamma}^*(S,\theta)\right)}{\theta:S\mapsto [0,1)}.
\end{equation}
By definition of $C_{k,\gamma}^*$ above we have that
\begin{eqnarray}
  \lefteqn{\Exp{v\left(C_{k,\gamma}^*(S,\theta)\right)}{\theta:S\mapsto [0,1)}} \nonumber
  \\ &=&
    \Exp{v\left(C_{k,\gamma}^{*[0,1/2)}(S,\theta)\right)}{\theta:S\mapsto [0,1)}
    + \Exp{v\left(C_{k,\gamma}^{*[1/2,1)}(S,\theta)\right)}{\theta:S\mapsto [0,1)}\nonumber
  \\ &=&
    \Exp{v\left(C_{k,\gamma}^{[0,1/2)}(S,\theta)\right)}{\theta:S\mapsto [0,1)}
    + \Exp{v\left(C_{k,\gamma}^{[1/2,1)}(S\cup D,\theta')\right)}{\theta':(S\cup D)\mapsto [0,1)} \ \ \  \label{eq:induction}
\end{eqnarray}

The mapping $\theta'$ in Equation (\ref{eq:induction}), with domain $S\cup D$, combines $\theta:S\mapsto [0,1)$ and $\theta_D:D\mapsto [0,1)$, this is well defined because $S$ and $D$ are disjoint.

We give a lower bound for the first term in Equation (\ref{eq:induction}) using the inductive hypothesis.  Using Lemma \ref{lem:second-half} we derive a lower bound for the 2nd term in Equation (\ref{eq:induction}).

To simplify the notation hereinafter we abbreviate $ \Exp{v(C_{k,\gamma}(S, \theta))}{\theta:S \mapsto [0,1)}$ as
$\Exp{C_{k,\gamma}(S)}$.

We first give a lower bound on  the first term in Equation (\ref{eq:induction}).
Given a set of items $S$, fix the set of items $S^{[0,1/2)}\subseteq S$ arriving in $[0,1/2)$. The arrival
times of $S^{[0,1/2)}$ are uniform in $[0,1/2)$. Therefore  conditioned on $S^{[0,1/2)}$
arriving in $[0,1/2)$, the expected profit of $C_{k,\gamma}$ from $S^{[0,1/2)}$
equals
\begin{equation*}
\Exp{C_{\floor{k/2},2\gamma}\left(S^{[0,1/2)}\right)}.
\end{equation*}

By induction we obtain
\begin{eqnarray*}
\Exp{C_{\floor{k/2},2\gamma}\left(S^{[0,1/2)}\right)}
&\ge& \frac{1}{1+2\gamma \floor{k/2}}
\bigr(1 - \beta(2\gamma, \floor{k/2})\bigr)\cdot T_{\floor{ k/2 }}\left(S^{[0,1/2)}\right) \ .
\end{eqnarray*}

It therefore follows that the expected profit of algorithm
$C_{k,\gamma}$ from elements arriving in $[0,1/2)$ (without any
conditioning on $S^{[0,1/2)}$), is
\begin{eqnarray}
\lefteqn{\Exp{C_{k,\gamma}^{[0,1/2)}(S)}} \nonumber
\\ &=&\Exp{\Exp{C_{\floor{k/2},2\gamma}\left(S^{[0,1/2)}\right)}}{S^{[0,1/2)}\subseteq S} \nonumber
\\ &\geq&\Exp{\frac{1}{1+2\gamma \floor{ k/2 }}
\bigr(1 - \beta(2\gamma, \floor{k/2})\bigr)\cdot T_{\floor{ k/2 }}\left(S^{[0,1/2)}\right)}{S^{[0,1/2)}\subseteq S} \nonumber
\\ &=&\frac{1}{1+2\gamma \floor{ k/2 }}
\bigr(1 - \beta(2\gamma, \floor{k/2})\bigr)\cdot \Exp{T_{\floor{ k/2 }}\left(S^{[0,1/2)}\right)}{S^{[0,1/2)}\subseteq S} \nonumber
\\ &\geq& \frac{1}{1+\gamma k} \bigr(1 - \beta(2\gamma, \floor{k/2})\bigr)\cdot \Exp{T_{\floor{ k/2 }}\left(S^{[0,1/2)}\right)}{S^{[0,1/2)}\subseteq S} \ .
\label{eq:profithalf}
\end{eqnarray}
The set $S^{[0,1/2)}$  is a uniformly random subset of $S$, therefore applying Lemma
\ref{lem:Bobby1} we get
\begin{equation} \label{tamara}
\Exp{T_{\floor{ k/2 }}\left(S^{[0,1/2)}\right)}{S^{[0,1/2)}\subseteq S}
\ge \left(1-\frac{1}{2\sqrt{k}}\right)\frac{T_{k}(S)}{2} \ .
\end{equation}
By substituting the lower bound (\ref{tamara}) into Equation
(\ref{eq:profithalf}) we get that
\begin{eqnarray}
\lefteqn{\Exp{C_{k,\gamma}^{[0,1/2)}(S)}} \nonumber \\
&\geq&
  \frac{1}{1+\gamma k}
  \bigr(1 -\beta(2\gamma, \floor{k/2})\bigr)
  \left(1-\frac{1}{2\sqrt{k}}\right)\frac{T_{k}(S)}{2}
\nonumber \\
&\ge&
  \frac{1}{1+\gamma k}
  \left(1 -\beta(2\gamma, \floor{k/2})  -\frac{1}{2\sqrt{k}} \right)
  \frac{T_{k}(S)}{2}
\nonumber \\
&\stackrel{{\rm sub.}\ \beta}{=}&
  \frac{1}{1+\gamma k}
  \left(
    1
    -\theoremCCVal\sqrt{2\gamma\ln(1/(2\gamma))}
    - \left(\frac{\ETwoVal}{\sqrt{\floor{k/2}}} + \frac{1}{2\sqrt{k}}\right)
  \right)
  \frac{T_{k}(S)}{2}
\nonumber \\
&\geq&
  \frac{1}{1+\gamma k}
  \left(
    \frac{1}{2}
    -\frac{\sqrt{2}}{2}\cdot\theoremCCVal\sqrt{\gamma\ln(1/\gamma)}
    - \left(\frac{\sqrt{2}}{2}\cdot\frac{\ETwoVal}{\sqrt{k-1}} + \frac{1}{4\sqrt{k}}\right)
  \right)
  T_{k}(S)\ \ \ \ \ \
\label{eq:profithalffinal}
\end{eqnarray}

where the last inequality in this derivation follows since $\floor{ k/2 } \ge (k-1)/2$.

Now we give a lower bound on the second term of Equation (\ref{eq:induction}). By Lemma \ref{lem:second-half} (recall that $\size{S\cup D} \geq 3\ceil{k/2}$) we obtain that the expected
profit of $C_{k,\gamma}$ executed on input $S\cup D$ during the time interval $[1/2, 1)$ is at least
\begin{eqnarray}
\lefteqn{\Exp{C_{k,\gamma}^{[1/2, 1)}(S\cup D)}} \nonumber \\
&\geq&
  \frac{1}{1 + \gamma k}
  \left(1 - \sqrt{2}\cdot\alpha(\gamma) - 4.5\gamma\right)
  \left(1 - \frac{2\sqrt{1 + 1/k}}{\sqrt{k}} - \frac{1}{k}\right)
  \frac{T_k(S\cup D)}{2} \nonumber\\
&=&
  \frac{1}{1 + \gamma k}
  \left(
    \frac{1}{2}
    - \left(\frac{\sqrt{2}}{2}\alpha(\gamma) + \frac{9}{4}\gamma\right)
    - \left(\frac{\sqrt{1 + 1/k}}{\sqrt{k}} + \frac{1}{2k}\right)\right)
  {T_k(S)}\label{eqn:second-half}\ ,
\end{eqnarray}

where $\alpha(\gamma) = \theoremC\sqrt{\gamma\ln(1/\gamma)}$ as define in Theorem \ref{THEOREM:HALF}. The last equality follows since $v(D) = 0$, therefore $T_k(S\cup D) = T_k(S)$.

Substituting the lower bounds (\ref{eq:profithalffinal}) and (\ref{eqn:second-half}) into Equation (\ref{eq:induction}) and separating terms that depend on $\gamma$ from those that depend on $k$ we get

\begin{eqnarray}\label{eqn:sumYZStar}
\Exp{C_{k, \gamma}^{*}(S)}
\geq
  \frac{T_k(S)}{1 + \gamma k}\cdot
    \Bigg(
      1
      &-& \left(\frac{\sqrt{2}}{2}\theoremCCVal\sqrt{\gamma\ln(1/\gamma)} + \frac{\sqrt{2}}{2}\alpha(\gamma) + \frac{9}{4}\gamma\right)
      \nonumber \\&-& \left(\frac{\sqrt{2}}{2}\cdot\frac{\ETwoVal}{\sqrt{k-1}} + \frac{1}{4\sqrt{k}} + \frac{\sqrt{1 + 1/k}}{\sqrt{k}} + \frac{1}{2k}\right)
    \Bigg)\ \ \ \ \ \
\end{eqnarray}

We observe that for all $\gamma \le \gammaUP = \gammaUPVal$
we have
\begin{eqnarray*}
    \lefteqn{\frac{\sqrt{2}}{2}\theoremCCVal\sqrt{\gamma\ln(1/\gamma)} + \frac{\sqrt{2}}{2}\alpha(\gamma) + \frac{9}{4}\gamma} \\
    &=& \left(\frac{\sqrt{2}}{2}\theoremCCVal + \frac{\sqrt{2}}{2}\cdot\theoremC + \frac{9}{4}\sqrt{\frac{\gamma}{\ln(1/\gamma)}}\right)\sqrt{\gamma\ln(1/\gamma)}
    \\&\leq& \left(\frac{\sqrt{2}}{2}\theoremCCVal + \frac{\sqrt{2}}{2}\cdot\theoremC + \frac{9}{4}\sqrt{\frac{\gammaUP}{\ln(1/\gammaUP)}}\right)\sqrt{\gamma\ln(1/\gamma)}
    \\&\approx& 7.3407561 \sqrt{\gamma\ln(1/\gamma)} < \theoremCC\sqrt{\gamma\ln(1/\gamma)} \ .
\end{eqnarray*}
Similarly for $k > 25$
we have
\begin{eqnarray*}
    \lefteqn{\frac{\sqrt{2}}{2}\cdot\frac{\ETwoVal}{\sqrt{k-1}} + \frac{1}{4\sqrt{k}} + \frac{\sqrt{1 + 1/k}}{\sqrt{k}} +\frac{1}{2k}} \\
    &=& \left(\frac{\sqrt{2}}{2}\frac{\ETwoVal}{\sqrt{1 - 1/k}} + \frac{1}{4} + \sqrt{1 + 1/k} + \frac{1}{2\sqrt{k}}\right)\frac{1}{\sqrt{k}}
    \\ &<& \left(\frac{\sqrt{2}}{2}\frac{\ETwoVal}{\sqrt{1 - 1/25}} + \frac{1}{4} + \sqrt{1 + 1/25} + \frac{1}{2\sqrt{25}}\right)\frac{1}{\sqrt{k}}
    \\ &\approx& \frac{4.9783}{\sqrt{k}} < \frac{\ETwoVal}{\sqrt{k}} \ .
\end{eqnarray*}

Substituting (\ref{eqn:sumYZStar}) into (\ref{eq:sizeValueCmp}) we obtain the statement of the theorem.
\qed\end{proof}

\begin{lemma}\label{lem:sizeDoesntMatter}

For any set $S$
$$\Exp{v\left(C_{k,\gamma}(S,\theta)\right)}{\theta:S\mapsto [0,1)}
\geq \Exp{v\left(C_{k,\gamma}^*(S,\theta)\right)}{\theta:S\mapsto [0,1)}.$$
\end{lemma}
\begin{proof}
Since $C_{k,\gamma}^{*}$ runs exactly as $C_{k,\gamma}$ in the time interval $[0,1/2)$ the expected value of both of them in this interval will be the same.
Hence all we need to prove is
\begin{equation}\label{eq:lemSecondInterval}
\Exp{v\left(C_{k,\gamma}^{[1/2,1)}(S,\theta)\right)}{\theta:S\mapsto [0,1)}
\geq \Exp{v\left(C_{k,\gamma}^{*[1/2,1)}(S,\theta)\right)}{\theta:S\mapsto [0,1)}.
\end{equation}

The proof splits the probability space $\theta:\mathrm{S}\mapsto [0,1)$ into subspaces.
We prove that Inequality (\ref{eq:lemSecondInterval}) holds for each subspace
and therefore it holds for the entire probability space.

We break the probability space by conditioning on the following:
\begin{itemize}
    \item  The subset of items of $S$ arriving in $[0,1/2)$.
We denote this subset by
 $S^{[0,1/2)}$. Note that once we condition on $S^{[0,1/2)}$ the following are also fixed
    \begin{itemize}
        \item The subset $S^{[1/2, 1)} \subseteq S$ of the items of $S$ arriving in  $[1/2,1)$.
        \item The threshold $T$ that  $C_{k, \gamma}$ use to pick items from $S^{[1/2, 1)}$.
        \item The subset $S_{\mbox{\tiny{\textgreater T}}}^{[1/2, 1)} \subseteq S^{[1/2, 1)}$ of items in $S_{\mbox{\tiny{\textgreater T}}}^{[1/2, 1)}$ with value at least $T$.
    \end{itemize}
    \item The time $t_{\ell}$ in which $C_{k,\gamma}$ last chose an element in $[0,1/2)$. Note that
given this conditioning $C_{k,\gamma}$ (and $C^*_{k,\gamma}$) can first pick an item in $[1/2,1)$  after time $\max{\set{1/2,t_{\ell}+\gamma }}$.
We denote this time by $\delta$.
\item The number of elements, $\lambda$, picked by $C_{k,\gamma}$ (and $C^*_{k,\gamma}$) in $[0,1/2)$.
    \item The arrival times $x_1, \ldots, x_q \in [1/2,1) $, $q = \size{S_{\mbox{\tiny{\textgreater T}}}^{[1/2, 1)}}$,  of the items
in $S_{\mbox{\tiny{\textgreater T}}}^{[1/2, 1)}$. (Note that we do not fix
which  item of $S_{\mbox{\tiny{\textgreater T}}}^{[1/2, 1)}$ arrives  at $x_i$ for any $1\le i \le q$.) We denote the
set $\set{x_i \mid x_i \ge \delta }$ by $X_\delta$.
\end{itemize}

This conditioning determines the  $\gamma$-independent subset $\Gamma$ of  $X_\delta$ in which $C_{k,\gamma}$ picks
items of $S_{\mbox{\tiny{\textgreater T}}}^{[1/2, 1)}$.
Let $m_\gamma(X_\delta)$ denote the size of the maximum $\gamma$-independent subset of  $X_\delta$.
The size of $\Gamma$ is the minimum of  $k-\lambda$ and $m_\gamma(X_\delta)$.
   The  expectation of $\Exp{v\left(C_{k,\gamma}^{[1/2,1)}(S,\theta)\right)}{\theta:S\mapsto [0,1)}$ in a subspace defined by the conditioning above is the average
over all 1-1 mappings of $S_{\mbox{\tiny{\textgreater T}}}^{[1/2,1)}$ to  the arrival times $x_1, \ldots, x_q$, of the values of the items
mapped to $\Gamma$.
In the following we denote this conditional expectation by $\Exp{C^{[1/2, 1)}_{k, \gamma}(S)}$. Hence we  get that
\begin{eqnarray}
\Exp{C^{[1/2, 1)}_{k, \gamma}(S)}
&=& \frac{\size{\Gamma}}{q} v(S_{\mbox{\tiny{\textgreater T}}}^{[1/2,1)}) \nonumber
\\ &=& \frac{\min{\set{k-\lambda,~ m_\gamma(X_\delta)}}}{q}\cdot v(S_{\mbox{\tiny{\textgreater T}}}^{[1/2,1)})\ \ . \label{eq:originalProfit}
\end{eqnarray}

Our conditioning does not fix the $\gamma$-independent subset $\Gamma^*$
of  $X_\delta$ in which $C^*_{k,\gamma}$ picks
items of $S_{\mbox{\tiny{\textgreater T}}}^{[1/2, 1)}$ because $\Gamma^*$ depend on $\theta_D$.
We denote the  expectation of $\Exp{v\left(C_{k,\gamma}^{*[1/2,1)}(S,\theta)\right)}{\theta:S\mapsto [0,1)}$ in a subspace defined by the conditioning above by $\Exp{C^{*[1/2, 1)}_{k, \gamma}(S)}$. As before we have
\begin{equation}\label{eq:fakeProfit}
\Exp{C^{*[1/2, 1)}_{k, \gamma}(S)} = \Exp{\frac{\size{\Gamma^*}}{q} v(S_{\mbox{\tiny{\textgreater T}}}^{[1/2,1)})} =  \frac{v(S_{\mbox{\tiny{\textgreater T}}}^{[1/2,1)})}{q}\Exp{\size{\Gamma^*}}
\end{equation}
 where the expectation of $\Gamma^*$
is taken over all 1-1 mappings of $S_{\mbox{\tiny{\textgreater T}}}^{[1/2,1)}$ to  the arrival times $x_1, \ldots, x_q$ and over $\theta_D$.

We can give an upper bound on the value of  $\size{\Gamma^*}$ for every
1-1 mappings of $S_{\mbox{\tiny{\textgreater T}}}^{[1/2,1)}$ to  the arrival times $x_1, \ldots, x_q$ and every $\theta_D$.
Note that $C^{*}_{k, \gamma}$ can select at most $k$ items and therefore $\size{\Gamma^*} \leq k-\lambda$.
It also holds that  $\Gamma^*$ is a $\gamma$-independent subset of $X_\delta$ and
therefore  $\size{\Gamma^*}\leq m_\gamma(X_\delta)$ hence
\begin{equation}\label{eq:GammaBound}
\size{\Gamma^*} \leq \min{\set{k-\lambda,~ m_\gamma(X_\delta)}}
\end{equation}

Substituting (\ref{eq:GammaBound}) into (\ref{eq:fakeProfit}) and comparing it to (\ref{eq:originalProfit}) we get that
$\Exp{C^{[1/2, 1)}_{k, \gamma}(S)} \geq \Exp{C^{*[1/2, 1)}_{k, \gamma}(S)}$.
Since this holds for any
subspace define by our conditioning on the
 values of $S^{[0,1/2)}$, $\delta$, $\lambda$ and $\set{x_1, \ldots, x_q}$ the lemma follows.

\qed\end{proof}

\begin{lemma} \label{lem:Bobby1}
Let $S$ be a set of size $n$ and let $Y$ be a subset of $S$ chosen uniformly at random
amongst all subsets of $S$. Then for any $25 < k \leq n$:

$$\Exp{T_{\floor{\half{k}}}(Y)} \geq \left(1-\frac{1}{2\sqrt{k}}\right)\cdot\frac{T_{k}(S)}{2}$$
\end{lemma}
\begin{proof}
From Lemma \ref{lem:randomsubsets} we know that
$$\Exp{T_{\floor{\half{k}}}(Y)} \geq T_{k}(S)\sum_{r=1}^{k}{{k \choose
r}\frac{1}{2^k}} \cdot\frac{\min(r, \floor{ k/2})}{k}\ .$$
 We consider even and odd $k$.

For $k$ even:
\begin{eqnarray}
\lefteqn{T_{k}(S)\sum_{r=1}^{k}{{k \choose r}\frac{1}{2^k}\cdot\frac{\min\left(r, \floor{ {\frac{k}{2}} }\right)}{k}}} \nonumber\\
&=&
  T_{k}(S)
  \left(
    \sum_{r=1}^{k/2}{{k \choose r}\frac{r}{k}\cdot\frac{1}{2^k}}
    + \sum_{r=k/2+1}^{k}{{k \choose r}\frac{1}{2^k}\cdot\frac{1}{2}}
  \right) \nonumber\\
&=&
  T_{k}(S)
  \left(
    \frac{1}{2}\sum_{r=1}^{k/2}{{k - 1 \choose r - 1}\frac{1}{2^{k-1}}}
    + \frac{1}{2}\cdot\prob{Bin\left(k, \frac{1}{2}\right) \geq \frac{k}{2}+1}
  \right)\nonumber\\
&=&
  \frac{T_{k}(S)}{2}
  \left(
    \sum_{r=0}^{k/2 - 1}{{k - 1\choose r}\frac{1}{2^{k-1}}}
    + \prob{Bin\left(k, \frac{1}{2}\right) \geq \frac{k}{2}+1}
  \right) \nonumber\\
&=&
  \frac{T_{k}(S)}{2}
  \left(
    \prob{Bin\left(k-1, \frac{1}{2}\right) \leq \frac{k}{2}-1}
    + \prob{Bin\left(k, \frac{1}{2}\right) \geq \frac{k}{2}+1}
  \right) \nonumber \\
&=&
  \frac{T_{k}(S)}{2}
  \left(
    \frac{1}{2}
    + \frac{1}{2} - {k \choose k/2}\frac{1}{2^{k+1}}
  \right) \label{eq:first_binomial332} \\
 &=& \frac{T_{k}(S)}{2}\left(1 - {k \choose k/2}\frac{1}{2^{k+1}}\right).\label{eq:first_binomial33}
\end{eqnarray}
We derive (\ref{eq:first_binomial332}) as follows:
\begin{itemize}
  \item
      $\prob{Bin\left(k-1, \frac{1}{2}\right) \leq \frac{k}{2}-1} = 1/2.$

      Consider $k-1$ coin tosses of a fair coin, the number of heads can be $0, 1, \ldots, k-1$.
      The probability that there are $i$ heads is equal to the probability that there are $k-1-i$ heads.
      As $k-1-(k/2-1)=k/2$, the set $\set{0, \ldots k/2-1 }$ is disjoint from the set  $\set{k/2,\ldots, k}$ and their union is $\set{0,\ldots,k-1}$.
  \item
      $\prob{Bin\left(k, \frac{1}{2}\right) \geq \frac{k}{2}+1}=\frac{1}{2} - {k \choose k/2}\frac{1}{2^{k+1}}.$

      Consider $k$ tosses of a fair coin, the outcome can have $0,1, \ldots, k$ heads (an odd number of outcomes).
      As above, the probability that there be $i$ heads, $i=0,1, \ldots, k/2-1$ is equal to the probability that there be $k-i$ heads.
      Ergo, $\prob{\mbox{$\#$ heads $\in 0, \ldots, k/2-1$}}=\prob{\mbox{$\#$ heads $\in k/2+1, \ldots, k$}}$,
      and \begin{eqnarray*}
    \lefteqn{\prob{\mbox{$\#$ heads $\in 0, \ldots, k/2-1$}}} \\
    & & +\prob{\mbox{$\#$ heads $\in k/2+1, \ldots, k$}} \\
    & & +\prob{\mbox{$\#$ heads $= k/2$}}=1.
    \end{eqnarray*}
      As the probability that there are exactly $k/2$ heads is ${k \choose k/2}\frac{1}{2^{k}}$,
      solving for $\prob{\mbox{$\#$ heads $\in k/2+1, \ldots, k$}}$ gives the desired result.
\end{itemize}

From Stirling's formula we get:
\begin{equation}\label{eq:after_binon3}
  {k \choose k/2}\frac{1}{2^k}
  \leq
    \frac{e\sqrt{k}(\frac{k}{e})^{k}}{\left(\sqrt{2\pi}\sqrt{\frac{k}{2}}{\left(\left(\frac{k}{2}\right)/e\right)^\frac{k}{2}}\right)^2}
     \cdot \frac{1}{2^k} =
\frac{e\sqrt{k}}{\pi \cdot k} <
\frac{1}{\sqrt{k}} \ .
\end{equation}
By substituting (\ref{eq:after_binon3}) in (\ref{eq:first_binomial33}) we get

\begin{equation*}
T_{\floor{k/2}}(S)  \geq \frac{T_k(S)}{2}\left(1 - \frac{1}{2\sqrt{k}}\right) \ ,
\end{equation*}
finishing the proof for $k$ even.

For an odd  $k$:
\begin{eqnarray}
\lefteqn{
  T_k(S)
    \sum_{r=1}^{k}{{k \choose r}\frac{1}{2^k}\frac{\min(r, \floor{\frac{k}{2}})}{k}} }
  \nonumber\\
&=&
  T_k(S)
  \left(
    \sum_{r=1}^{\floor{\frac{k}{2}}}{{k \choose r}\frac{r}{k}\cdot\frac{1}{2^k} }
    + {\floor{\frac{k}{2}}}\cdot\frac{1}{k}\sum_{r=\floor{\frac{k}{2}}+1}^{k}{{k \choose r}\frac{1}{2^k}}
  \right) \nonumber \\
&\geq&
  T_k(S)
  \left(
    \frac{1}{2}\sum_{r=1}^{\floor{\frac{k}{2}}}
      {{k-1 \choose r-1}\frac{1}{2^{k-1}}}
    ~+~ \frac{k-1}{2k}\sum_{r=\floor{\frac{k}{2}}+1}^{k}{{k \choose r}\frac{1}{2^k}}
  \right) \nonumber \\
&=&
  \frac{T_k(S)}{2}
  \left(
    \prob{Bin\left(k-1, \frac{1}{2}\right) \leq \floor{\frac{k}{2}} - 1}
    + \frac{k-1}{k}\prob{Bin\left(k, \frac{1}{2}\right) \geq \floor{\frac{k}{2}} + 1}
  \right) \nonumber\\
&=&
  \frac{T_k(S)}{2}
  \left(
    \left(\frac{1}{2} - {{k-1} \choose (k-1)/2}\frac{1}{2^k}\right)
    + \frac{k-1}{k} \cdot \frac{1}{2}
  \right) \nonumber\\
&=&
\frac{T_k(S)}{2}\left(1 - \left({{k-1} \choose (k-1)/2}\frac{1}{2^k} + \frac{1}{2k}\right)\right). \label{eqn:something43}
\end{eqnarray}
Lemma \ref{lem:induction111} shows by induction that for any odd $k > 25$
\begin{equation}\label{eqn:something45}
{k-1 \choose (k-1)/2}\frac{1}{2^k} + \frac{1}{2k} \leq \frac{1}{2\sqrt{k}} \ .
\end{equation}
(Note that by (\ref{eq:after_binon3}) we get that
$$
{k-1 \choose (k-1)/2}\frac{1}{2^k} + \frac{1}{2k} \le \frac{1}{2}\left(\frac{e}{\pi} \frac{1}{\sqrt{k-1}} + \frac{1}{k} \right) $$
which is smaller than $1/(2\sqrt{k})$ for $k \ge 63$.)

Substituting (\ref{eqn:something45}) into (\ref{eqn:something43}) we get
\begin{equation*}
T_{\floor{k/2}}(S)  \geq \frac{T_k(S)}{2}\left(1 - \frac{1}{2\sqrt{k}}\right) \ ,
\end{equation*}
finishing the proof for $k$ odd.

\qed\end{proof}

\begin{lemma} \label{lem:randomsubsets}
Let $S = \set{s_1 > s_2 > \dots > s_n}$ and let $S'$ be a subset of $S$ chosen uniformly at random
amongst all subsets of $S$. Let $t, k$ be integers s.t. $t\le k \le n$. Then\\
$$\Exp{[T_t(S')]} \geq T_{k}(S) \sum_{r=1}^{k}{{k \choose r}\frac{1}{2^{k}} \frac{\min(r, t)}{k} } \ . $$
\end{lemma}
\begin{proof}
Let $R = \set{s_1 > \ldots > s_k}$ be the $k$ elements of largest value in $S$.
Conditioned upon $\size{S'\cap  R} = r$, the expectation of
$v(S' \cap R)$ is  $r/k$ times the sum of the values in $R$ which is $T_k(S)$.

The lemma now follows by summing over all
possible values of $r$ using the following two facts:
\begin{enumerate} \item The probability  that
 $\size{S'\cap R}=r$ is ${k
\choose r}  \frac{1}{2^k}$.
\item  If $t\le r$ then $T_t(S')\ge T_t(S'\cap R) \ge \frac{t}{r} v(S' \cap R)$ and if $t\ge r$ than
$T_t(S')\ge  v(S' \cap R)$.
\end{enumerate}
\qed\end{proof}

\begin{lemma}\label{lem:induction111}
For any odd $k > 25$
$${k-1 \choose (k-1)/2}\frac{1}{2^{k}} + \frac{1}{2k} \leq \frac{1}{2\sqrt{k}}$$
\end{lemma}
\begin{proof}
We prove the lemma by induction on $k$

Note that this inequality doesn't hold for $k
\leq 25$.

For simplicity, we prove the equivalent inequality:
$$\left({k-1 \choose (k-1)/2}\frac{1}{2^k} + \frac{1}{2k}\right)\cdot2\sqrt{k} \leq 1$$
For the basis of the induction we verify this inequality for
$k=27$:
$$\left({k-1 \choose (k-1)/2}\frac{1}{2^k} + \frac{1}{2k}\right)\cdot {2\sqrt{k}} \approx 0.997755077 < 1.$$
We assume that the inequality holds for some odd $k$ and show that it holds for $k+2$. In the following derivation  Inequality (\ref{eqn:here_induction}) follows by the induction hypothesis.
\begin{eqnarray}
\lefteqn{\left({k+2-1 \choose (k+2-1)/2}\frac{1}{2^{k+2}} + \frac{1}{2(k+2)}\right)\cdot 2\sqrt{k+2}} \nonumber\\
&=& \left(\frac{k(k+1)}{((k+1)/2)^2} \cdot \frac{1}{4}\cdot{k-1 \choose (k-1)/2}\frac{1}{2^{k}} + \frac{1}{2(k+2)}\right)\cdot 2\sqrt{k+2} \nonumber\\
&\leq& \left(\frac{k}{k+1}\cdot\left(\frac{1}{2\sqrt{k}} - \frac{1}{2k}\right) + \frac{1}{2(k+2)}\right)\cdot 2\sqrt{k+2} \label{eqn:here_induction}\\
&=& \frac{\sqrt{k}\sqrt{k+2}}{k+1} - \frac{\sqrt{k+2}}{k+1} + \frac{1}{\sqrt{k+2}} \nonumber\\
&=& \frac{\sqrt{k}\sqrt{k+2}}{k+1} + \frac{k+1 - (k+2)}{(k+1)\sqrt{k+2}} \nonumber\\
&=& \frac{\sqrt{k^2 + 2k + 1 - 1}}{k+1} - \frac{1}{(k+1)\sqrt{k+2}} \nonumber\\
&=& \frac{\sqrt{(k+1)^2 - 1}}{k+1} - \frac{1}{(k+1)\sqrt{k+2}} \nonumber\\
&<& 1 + 0  \nonumber\\
&=& 1 \nonumber
\end{eqnarray}
\qed\end{proof}

\noindent{\bf Definitions:}
The  following definitions are used in Lemmata
    (\ref{lem:second-half}) -- (\ref{lem:notSimpleArithmetics2}).

\begin{itemize}
\item
    Let  $S^{[0,1/2)}=\set{y_1 > y_2 >  \cdots > y_{m}} \subseteq S$ be  random variables for the set of items arriving during the time period $[0,1/2)$,
\item
    and let  $S^{[1/2 , 1)}=\set{z_1 > z_2 >  \cdots > z_{\size{S} - m}} \subseteq S$ be  random variables for the set of items arriving during the time period
    $[1/2,1)$.
\item
    Let $Q = \size{\set{z \in S^{[1/2, 1)} \mid z > T}}$ be the number of elements in $S^{[1/2, 1)}$ greater than the threshold $T$.
    In the Algorithm~\ref{algo:tempalg} we define the threshold $T$ where $T=y_\tau$, $\tau=\ceil{k/2}$,  if $\size{S^{[0,1/2)}} \ge \tau$ and
    $T=0$ otherwise (in this case we consider any item of value $0$ as greater than $T$).
\item
    Let $G_i$, $1 \leq i \leq \tau$,  be identical independent  geometric random
    variables, such that for any integer $j\ge 0$, $\prob{G_i = j} =
    \left(\frac{1}{2}\right)^{j+1}$. It follows that the expectation of $G_i$, $\Exp{G_i} = 1$, whereas the variance $\sigma^2[G_i] = 2$. Let $G =
    \sum_{i = 1}^{\tau}{G_i}$. Note also that $\Exp{G} = \tau$ and $\sigma^2[G] = 2\tau$.
\item
    We abbreviate $ \Exp{v(C_{k,\gamma}(S, \theta))}{\theta:S \mapsto [0,1)}$ as
    $\Exp{C_{k,\gamma}(S)}$ as in the proof of Theorem \ref{THEOREM:MAIN}.
\end{itemize}

\begin{lemma} \label{lem:second-half}
Given that $\size{S} \geq 3\ceil{k/2}$, for any $\gamma \leq \gammaUP = \gammaUPVal$ and any $k \leq 1/\gamma$, the expected profit of $C_{k,\gamma}$ during time interval $[1/2,1)$
is at least
$$
\Exp{C_{k, \gamma}^{[1/2, 1)}(S)} \geq
  \frac{1}{1 + \gamma k}
  \left(1 - \sqrt{2}\cdot\alpha(\gamma) - 4.5\gamma\right)
  \left(1 - \frac{2\sqrt{1 + 1/k}}{\sqrt{k}} - \frac{1}{k}\right)\frac{T_k(S)}{2} \ ,
$$

where $\alpha(\gamma) = \theoremC\sqrt{\gamma\ln(1/\gamma)}$ as defined in Theorem \ref{THEOREM:HALF}.

\end{lemma}
\begin{proof}

Splitting the probability space by conditioning on $Q$ we obtain
\begin{eqnarray}
\Exp{C_{k, \gamma}^{[1/2, 1)}(S)}
&=& \sum_{q=0}^{\infty}{\prob{Q = q}\cdot{\Exp{C_{k, \gamma}^{[1/2, 1)}(S) \mid Q = q}}} \nonumber \\
&\geq& \sum_{q = 0}^{2\ceil{k/2}}
    {\prob{Q = q}\cdot{\Exp{C_{k, \gamma}^{[1/2, 1)}(S) \mid Q =
    q}}} \label{eq:condQq} \ .
\end{eqnarray}

Lemma \ref{lemma:prof_r} shows that
${\Exp{C_{k, \gamma}^{[1/2, 1)}(S) \mid Q = q}} \ge \frac{1}{1 + \gamma k}
        \left(1 - \sqrt{2}\cdot\alpha(\gamma) - 4.5\gamma\right)\cdot
        \left(1 - 2\frac{|q - \ceil{k/2}|}{k} - \frac{1}{k}\right)
        \cdot \frac{T_k(S)}{2}$.
       Substituting this lower bound in (\ref{eq:condQq}) we obtain

$$
\begin{array}{l}
\Exp{C_{k, \gamma}^{[1/2, 1)}(S)}
 \geq
  \sum_{q = 0}^{2\ceil{k/2}} \Bigg(\prob{Q = q}\cdot \bigg[\frac{1}{1 + \gamma k}
        \bigg(1 - \sqrt{2}\cdot\alpha(\gamma) - 4.5\gamma\bigg)  \\
\hspace{2.5cm}           \cdot
        \bigg(1 - 2\frac{|q - \ceil{k/2}|}{k} - \frac{1}{k}\bigg)
        \cdot \frac{T_k(S)}{2}
      \bigg]\Bigg)    \\
\hspace{2cm} =
\frac{1}{1 + \gamma k} \cdot
  \left(1 - \sqrt{2}\cdot\alpha(\gamma) - 4.5\gamma\right)
   \frac{T_k(S)}{2} \nonumber \\
\hspace{2.5cm} \cdot  \sum_{q = 0}^{2\ceil{k/2}}
    {
      \prob{Q = q}
      \left(1 - 2\frac{|q - \ceil{k/2}|}{k} - \frac{1}{k}\right)
    }  \label{eqn:geom}
\end{array}$$

Lemma
\ref{lem:geometric_equal}  proves that if $|S|\geq 3 \ceil{k/2}$ then for any $q \leq 2\ceil{k/2}$ the probability that
$Q=q$ is the same as the probability that the sum of $\ceil{k/2}$ identical geometrical random variables
is equal $q$. We denote this sum by $G$ and derive the following lower bound on the last sum of Equation (\ref{eqn:geom}):

\begin{eqnarray}
\lefteqn{\sum_{q = 0}^{2\ceil{k/2}}
    {
      \prob{Q = q}
      \left(1 - 2\frac{|q - \ceil{k/2}|}{k} - \frac{1}{k}\right)
    }
} \nonumber \\
&=&
  \sum_{q = 0}^{2\ceil{k/2}}
    {
      \prob{G = q}
      \left(1 - 2\frac{|q - \ceil{k/2}|}{k} - \frac{1}{k}\right)
    } \nonumber \\
&=&
  \sum_{r = 0}^{\ceil{k/2}}
    {\prob{|G - \ceil{k/2}| = r}\cdot
      \left(1 - 2\cdot\frac{r}{k} - \frac{1}{k}\right)
    } \label{eq:QExplain1} \\
&\geq&
  \sum_{r = 0}^{\infty}
    {\prob{|G - \ceil{k/2}| = r}\cdot
      \left(1 - 2\cdot\frac{r}{k} - \frac{1}{k}\right)
    } \label{eq:QExplain2} \\
&=& 1 - 2\cdot\frac{\Exp{|G -\ceil{\frac{k}{2}}|}}{k}  -
\frac{1}{k} \nonumber \\
&=& 1 - 2\cdot\frac{\Exp{|G -\mathrm{E}[G]|}}{k}  -
\frac{1}{k} \label{eq:Q} \ \ .
\end{eqnarray}
Equality (\ref{eq:QExplain1}) follows by the change of variables $r= |q -\ceil{k/2}|$, Inequality (\ref{eq:QExplain2}) follows
since $\left(1 - 2\cdot\frac{r}{k} - \frac{1}{k}\right) < 0$ for any $r > \ceil{k/2}$ and the last equality follows since $\ceil{k/2} = E(G)$.

Using the fact that $\Exp{|W - \Exp{W}|} \le \sqrt{\mathrm{Var}[W]}$ ($\mathrm{Var}$ here stands for  variance) for any random variable $W$ (Lemma \ref{lem:random_dist})
and the fact that $\mathrm{Var}[G] = 2\ceil{k/2} \le k+1$ we get
\begin{equation} \label{eq:var}
\frac{\Exp{|G - \mathrm{E}[G]|}}{k} \leq
 \frac{\sqrt{k+1}}{k} =
 \frac{\sqrt{1 + 1/k}}{\sqrt{k}}.
\end{equation}

Substituting the upper bound of Equation (\ref{eq:var}) into Equation (\ref{eq:Q}), and then substituting the resulting inequality into Equation
(\ref{eqn:geom}) we obtain the lemma.

\qed\end{proof}


It may be useful to observe that lower bound given in the following Lemma \ref{lemma:prof_r}, is approximately equal to the expected size of the  $\gamma$-independent set chosen from the $q$ items arriving during the time interval $[1/2,1)$ divided by $q$, multiplied by the expected value of these $q$ elements ($\left(1 - 2\frac{|q - \ceil{k/2}|}{k} - \frac{1}{k}\right)
        \cdot \frac{T_k(S)}{2}$).

\begin{lemma}\label{lemma:prof_r}
For any $\gamma \leq \gammaUP = \gammaUPVal$, any $k\leq 1/\gamma$, any $S$, $|S|\geq k$, and for any $0 \leq q \leq 2\ceil{k/2}$ we have
    \begin{eqnarray*}
    \lefteqn{\Exp{C_{k, \gamma}^{[1/2, 1)}(S) \mid Q = q}} \\
    &&\geq \frac{1}{1 + \gamma k}\left(1 - \sqrt{2}\alpha(\gamma) -
    4.5\gamma\right)\left(1 - 2\cdot\frac{|q - \ceil{k/2}|}{k}-\frac{1}{k}\right)
    \frac{T_k(S)}{2}\ \ .
    \end{eqnarray*}
\end{lemma}
\begin{proof}

We condition  upon there being $q$ items exceeding the threshold in $S^{[1/2,1)}$. Let $S_{>T}$ be the set of all items in $S$ greater than $T$.

Assume $|S^{[0,1/2)}|\geq \tau$ ({\sl i.e.}, sufficiently many items arrive before time $1/2$ so as to take $y_\tau$ as a threshold).  Therefore there are exactly $\tau-1$ items $> y_\tau$ in $S^{[0,1/2)}$, and from the conditioning there are exactly $q$ items $> y_{\tau}$ in $S^{[1/2,1)}$, so  $|S_{>T}|=\tau-1+q$ is the number items in $S$ that are strictly greater than the threshold $y_\tau$. Thus, the probability that an item $x\in S_{>T}$  is in $S^{[1/2,1)}$ is
\begin{equation*}
\prob{x \in S^{[1/2, 1)} \mid Q = q} = \frac{q}{|S_{>T}|}=\frac{q}{\tau + q - 1}\ .
\end{equation*}

If $|S^{[0,1/2)}|< \tau$, the threshold is set to be zero, and  $S=S_{>T}$. The size of $S$ is $|S^{[0,1/2)}|+q< \tau+q$.
It follows that the probability that an item $x\in S_{>T}$  is in $S^{[1/2,1)}$ is
\begin{equation*}
\prob{x \in S^{[1/2, 1)} \mid Q = q} =  \frac{q}{|S_{>T}|} \geq \frac{q}{\tau + q - 1}\ .
\end{equation*}

We now consider two cases: $q\leq\ceil{k/2}$, and $\ceil{k/2} < q \leq 2\ceil{k/2}$.

For $q\leq\ceil{k/2}$ the expected sum of the values of the $q$ items in $S^{[1/2,1)}\cap S_{>T}$ is
\begin{eqnarray}
\Exp{v(\set{x \mid x \in S^{[1/2,1)}\cap S_{>T}}) \mid Q = q}
&=& \sum_{x\in S_{>T}}{\prob{x \in S^{[1/2, 1)} \mid Q = q}}\cdot v(x)\nonumber
\\&=& \frac{q}{|S_{>T}|} \cdot \sum_{x\in S_{>T}} v(x)\nonumber
\\&\geq& \frac{q}{k}\cdot T_k(S) \label{eqn:partial1}
\end{eqnarray}

Equation (\ref{eqn:partial1}) follows since $|S_{>T}|\leq\tau - 1+q \leq \ceil{k/2}-1+\ceil{k/2} \leq k$ and $T_k(S)/k$ (the average value of the $k$ largest items) is smaller or equal to the average value of the items greater than $T$, ($\sum_{x\in S_{>T}} v(x)/|S_{>T}|$).

Since $q \leq \ceil{k/2}$, $C_{k, \gamma}$ chooses a $\gamma$-independent subset
of size at least as large as the size of the maximum $\gamma$-independent subset of those among
the $q$ items arriving in the interval $[1/2+\gamma,1)$.

Fixing the arrival times of the $q$ items in $S^{[1/2, 1)}$ which are above the threshold, the assignment of items to these times is a random permutation.
This implies that the expected value of the items in the $\gamma$-independent subset picked by $C_{\gamma,k}$ equals to the cardinality of this
$\gamma$-independent subset divided by $q$, multiplied by the expected value of these items.
Lemma (\ref{lem:half_gamma_indep}) gives a  lower bound on the expected size of the $\gamma$-independent subset
picked by $C_{k, \gamma}$ and Equation (\ref{eqn:partial1}) gives a lower bound on the expected value of the $q$ items, therefore:
\begin{eqnarray}
\lefteqn{\Exp{C_{k, \gamma}^{[1/2, 1)}(S) \mid Q = q}} \nonumber \\
&\geq&
  \frac{1}{q}\cdot\left[\frac{q}{1 + 2\gamma q}\left(1 - \sqrt{2}\alpha(\gamma) -4\gamma\right)\right]
  \cdot
  \left[\frac{q}{k}\cdot T_k(S)\right]\label{eq:prevExample} \\
&\geq&
  \frac{1}{1 + \gamma (k+1/2)}
  \left(1 - \sqrt{2}\alpha(\gamma) - 4\gamma\right)
  \cdot\left[\frac{q}{k}\cdot T_k(S)\right] \nonumber \\
&=&
  \frac{1}{1+\gamma k} \cdot \frac{1+\gamma k}{1 + \gamma (k+1/2)}
  \left(1 - \sqrt{2}\alpha(\gamma) - 4\gamma\right)
  \left[\frac{q}{k}\cdot T_k(S)\right] \label{eq:chain1} \\
&=&
  \frac{1}{1+\gamma k} \cdot \frac{1}{1 + \gamma/2 \cdot 1/(1+\gamma k))}
  \left(1 - \sqrt{2}\alpha(\gamma) - 4\gamma\right)
  \left[\frac{q}{k}\cdot T_k(S)\right]\nonumber \\
&\geq&
  \frac{1}{1+\gamma k} \cdot \frac{1}{1 + \gamma/2}
  \left(1 - \sqrt{2}\alpha(\gamma) - 4\gamma\right)
  \left[\frac{q}{k}\cdot T_k(S)\right]\nonumber \\
&\geq&
  \frac{1}{1+\gamma k} \cdot (1-\gamma/2)
  \left(1 - \sqrt{2}\alpha(\gamma) - 4\gamma\right)
  \left[\frac{q}{k}\cdot T_k(S)\right]\label{eq:chain3} \\
&\geq&
  \frac{1}{1+\gamma k}
  \left(1 - \sqrt{2}\alpha(\gamma) - 4.5\gamma\right)
  \left(1 - 2\cdot\frac{|q - \ceil{k/2}|}{k}\right)
  \frac{T_k(S)}{2}\nonumber \\
&\geq&
  \frac{1}{1+\gamma k}
  \left(1 - \sqrt{2}\alpha(\gamma) - 4.5\gamma\right)
  \left(1 - 2\cdot\frac{|q - \ceil{k/2}|}{k} - \frac{1}{k}\right)
  \frac{T_k(S)}{2}\nonumber
\end{eqnarray}
where Inequality \ref{eq:chain1} follow since $q \leq \ceil{k/2}$ and Inequality (\ref{eq:chain3}) follows since $\frac{1}{1-x} \geq 1+x$ for any $x$.
In addition we have that for any $\gamma \leq \gammaUP$
$$\left(1 - \sqrt{2}\alpha(\gamma) - 4.5\gamma\right) \geq \left(1 - \sqrt{2}\alpha(\gammaUP) - 4.5\gammaUP\right) > 0\ ,$$
therefore all terms in the previous calculations are positive meaning that decreasing any of them decreases the whole product.
This concludes the proof for $q \leq \ceil{k/2}$.

For $\ceil{k/2} < q \leq 2\ceil{k/2}$, the expected sum of the values of the $q$ items in $S^{[1/2,1)}\cap S_{>T}$ is
\begin{eqnarray}
\Exp{v(\set{x \mid x \in S^{[1/2,1)}\cap S_{>T}}) \mid Q = q}
&\geq& \sum_{x\in S_{>T}}{\prob{x \in S^{[1/2, 1)} \mid Q = q}}\cdot v(x) \nonumber
\\&=& \frac{q}{\size{S_{>T}}}\sum_{x\in S_{>T}} v(x)\nonumber
\\&\geq& \frac{q}{q + \tau - 1} \sum_{x\in S_{>T}} v(x)\nonumber \\
&\geq& \frac{q}{q + \tau - 1} T_k(S) \label{eqn:partial2} \ .
\end{eqnarray}
If $\size{S^{[0,1/2)}} < \tau$ the threshold is zero so $S_{>T} = S$ and therefore Inequality (\ref{eqn:partial2}) follows.
Otherwise $\size{S_{>T}} = \tau - 1 + q \geq \ceil{k/2} - 1 + \ceil{k/2} + 1\geq k$ so $S_{>T}$ contains the $k$ largest items of $S$, and therefore Inequality (\ref{eqn:partial2}) follows.

Let $s$ be the
size of the maximum $\gamma$-independent subset of the
the $q$ items larger than $T$ arriving in the interval $[1/2+\gamma,1)$.

Since $C_{k, \gamma}$ is restricted to choose at most $\ceil{k/2}$ items in the time interval $[1/2, 1)$ the size of the $\gamma$-independent subset that $C_{k,\gamma}$ chooses
is at least $\min(s,\ceil{k/2})$. Therefore,
\begin{eqnarray}
\Exp{\size{C_{k,\gamma}^{[1/2,1)}(S)} \mid Q=q}
&\geq&\Exp{\min\left(s, \ceil{k/2}\right)  \mid Q=q} \nonumber
\\&=& \Exp{\frac{s \cdot \ceil{k/2}}{\max(s, \ceil{k/2})}  \mid Q=q} \nonumber
\\&\geq& \Exp{s \mid Q=q}\cdot\frac{\ceil{k/2}}{q} \nonumber
\\&\geq&
\left(
  \frac{q}{1+2\gamma q}
  \left(
    1
    - \sqrt{2}\alpha(\gamma)
    - 4\gamma\right)
\right)
\cdot\frac{\ceil{k/2}}{q} \nonumber
\\&\geq&
\left(
  \frac{q}{1+2\gamma q}
  \left(
    1
    - \sqrt{2}\alpha(\gamma)
    - 4.5\gamma\right)
\right)
\cdot\frac{\ceil{k/2}}{q} \label{eq:blabla1}
\end{eqnarray}
where Equation (\ref{eq:blabla1}) follows from Lemma \ref{lem:half_gamma_indep}.

Combining Equations (\ref{eqn:partial2}) and (\ref{eq:blabla1}) using the same argument that we used to derive Equation (\ref{eq:prevExample}) we get
\begin{eqnarray}
\lefteqn{\Exp{C_{k, \gamma}^{[1/2, 1)}(S) \mid Q = q} } \nonumber \\
&& \geq \frac{1}{q} \cdot \left[\left(
  \frac{q}{1+2\gamma q}
  \left(
    1
    - \sqrt{2}\alpha(\gamma)
    - 4.5\gamma\right)
\right)\right]
\cdot \left[\frac{\ceil{k/2}}{q}
\cdot \frac{q}{q+\tau - 1}T_k(S) \right] \label{eqn:haba12}
\end{eqnarray}

A simple arithmetic manipulation (for more details see Lemma \ref{lem:notSimpleArithmetics1}) gives us
\begin{equation}\label{eqn:arithmeticsSub1}
\frac{\ceil{k/2}}{\tau + q - 1} \geq \frac{1}{2}\left(1 - \frac{q - \ceil{\half{k}}}{k}\right) \ .
\end{equation}
Substituting (\ref{eqn:arithmeticsSub1}) into (\ref{eqn:haba12}) we obtain
\begin{equation}\label{eqn:beforeArithmeticsSub2}
\Exp{C_{k, \gamma}^{[1/2, 1)}(S) \mid Q = q} \geq \frac{1}{1+2\gamma q}\left(1 - \sqrt{2}\alpha(\gamma) - 4.5\gamma\right) \cdot
\frac{1}{2}\left(1 - \frac{q - \ceil{\half{k}}}{k}\right)T_k(S)
\end{equation}
Another simple arithmetic manipulation (for more details see Lemma \ref{lem:notSimpleArithmetics2}) gives us
\begin{eqnarray}
\frac{1}{1+2\gamma q} \left(1 - \frac{q - \ceil{k/2}}{k}\right)
&\geq& \frac{1}{1+\gamma k}\left(1 - 2\cdot\frac{q-k/2}{k}\right) \label{eqn:preArithmeticsSub2}
\\&\geq& \frac{1}{1+\gamma k}\left(1 - 2\cdot\frac{|q-\ceil{k/2}|}{k} -\frac{1}{k}\right)\ . \label{eqn:arithmeticsSub2}
\end{eqnarray}
By substituting (\ref{eqn:arithmeticsSub2}) into (\ref{eqn:beforeArithmeticsSub2}) the lemma follows.
\qed\end{proof}

\begin{lemma}\label{lem:geometric_equal}
For any set $S$ and any $q \leq \size{S} - \tau$
$$
\prob{Q = q} = \prob{G = q}.
$$
($Q$ and $G$ are defined before \ref{lem:second-half}).
\end{lemma}
\begin{proof}

We reinterpret the random schedule as though we schedule the items in $S$ in decreasing order.

Repeatedly toss a fair coin until $\tau$ ``tails" appear. The length of this sequence is the sum of $\tau$ geometric variables $G_i$ (number of consecutive ``heads") plus $\tau$ (number of ``tails"), let $G=\sum_{i=1}^\tau G_i$ be a random variable for the total number of ``heads" in this sequence.

Traverse this sequence until its end or $S$ is exhausted, schedule the top item remaining in $S$ within the interval $[1/2,1)$ if ``heads", and within the interval $[0,1/2)$ if ``tails".

If $S$ has not been exhausted, place remaining items to appear at random times in time interval $[\,0,1)$.

If $S$ was  exhausted strictly before the end of the sequence --- we have that $G +\tau > \size{S}$ (total number of coin tosses),
and the number of items arriving in the time interval $[0,1/2)$ is $< \tau$, thus all items arriving during $[1/2, 1)$ must be counted in $Q$ and therefore $Q > \size{S} -\tau$ as well.

If $S$ was not exhausted strictly before the last coin toss, then $G +\tau \leq \size{S}$, ergo, any item placed in the time interval $[1/2, 1)$ before the last coin toss contributes to $Q$, and any other item must not be counted in $Q$
(either it was placed in time interval $[0,1/2)$ or its value is less than the threshold value --- the minimum of the $\tau$ largest items scheduled during the time interval $[0,1/2)$. Thus, $Q = G \leq \size{S} - \tau$.

We conclude that $G > \size{S} -\tau$ iff $Q > \size{S} -\tau$ and for any $q \leq \size{S} - \tau$,
  $G = q $ iff $Q = q $. So
\begin{equation*}
\prob{G > \size{S} - \tau} = \prob{Q > \size{S}-\tau},
\end{equation*}
and for any $q \leq \size{S} - \tau$
\begin{equation*}
\prob{G = q} = \prob{Q = q}.
\end{equation*}
\qed\end{proof}

\begin{lemma}\label{lem:random_dist}
For any random variable $W$, with mean $\mu$ and standard deviation
$\sigma$,
$$\mathrm{E}\bigr[|W - \mu|\bigr] \le \sigma$$
\end{lemma}
\begin{proof}
By the definition of the variance
\begin{eqnarray*}
\mathrm{Var}\bigr[|W - \mu|\bigr]
&=& \mathrm{E}\bigr[|W - \mu|^2\bigr] - \mathrm{E}^2\bigr[|W - \mu|\bigr] \\
&=& \mathrm{Var}\bigr[W\bigr] -
\mathrm{E}^2\bigr[|W - \mu|\bigr]
\end{eqnarray*}
Since $\mathrm{Var}\bigr[|W - \mu|\bigr]\geq 0$ it follows that
$$\mathrm{E}^2\bigr[|W - \mu|\bigr] \le \mathrm{Var}\bigr[W\bigr] = \sigma^2\ .$$
The lemma follows by taking the square root of both sides.
\qed\end{proof}

\begin{lemma}\label{lem:half_gamma_indep}
Let $0 < \gamma \leq 1$ and let $\Psi$ be a set of items scheduled at uniform times during the  interval
$[1/2, 1)$, and let $\psi=|\Psi|$. Let $\Psi'$ be the subset of $\Psi$ that are scheduled during the time interval $[1/2+\gamma,1)$. The expected size of the maximum $\gamma$-independent
subset of items from $\Psi'$ is at least
$$\frac{\psi}{1+ 2\gamma \psi}\left(1-\sqrt{2}\alpha(\gamma) - 4\gamma\right),$$
where $\alpha(\gamma) = \theoremC\sqrt{\gamma\ln(1/\gamma)}$ (as defined in Theorem \ref{THEOREM:HALF}).
\end{lemma}
\begin{proof}
If $\gamma \geq 1/2$ the lemma holds since the bound is negative, therefore assume $\gamma < 1/2$.

Let $\rho$ be the size of the maximum $\gamma$-independent subset of $\Psi$, and let $\rho'$ be the size of the maximal $\gamma$-independent subset of $\Psi'$ (both $\rho$ and $\rho'$ are random variables). It is easy to see that $\rho' \geq \rho-1$.

Let $I$ be an indicator variable where $I=1$ if some item in $\Psi$ is scheduled during the time interval $[1/2,1/2+\gamma)$, and $I=0$ otherwise.

\begin{eqnarray}
\Exp{\rho'} &\ge& \prob{I=0} \cdot \Exp{\rho \mid I=0} ~~+~~ \prob{I=1}\cdot\bigr(\Exp{\rho \mid I=1} - 1\bigr) \nonumber\\
&=& \Exp{\rho} - \prob{I=1}. \nonumber
\end{eqnarray}
To bound $\Exp{\rho}$ we recall that $\Psi$ consists of $\psi$ points randomly placed in the time interval $[1/2,1)$. Every  $\gamma$-independent subset of $\Psi$ is analogous to a $2\gamma$-independent subset of $\widetilde{\Psi}$ where $\widetilde{\Psi} = \{2p-1 \mid p\in \Psi\}$. {\sl I.e.}, the size of the maximal $\gamma$-independent subset of $\Psi$ is equal to the size of the maximal $2\gamma$-independent subset of $\widetilde{\Psi}$. Ergo, we can apply Theorem \ref{THEOREM:HALF} and hence:
\begin{eqnarray}
\Exp{\rho} - \prob{I=1}
&\geq& \frac{\psi}{1+ 2\gamma \psi}(1-\alpha(2\gamma)) - \prob{I=1} \nonumber\\
&=& \frac{\psi}{1+ 2\gamma \psi}\Bigr(1-\alpha(2\gamma)\Bigr) - (1 - (1-2\gamma)^\psi) \nonumber\\
&=& \frac{\psi}{1+ 2\gamma \psi}\left(1-\alpha(2\gamma) - \frac{(1 - (1-2\gamma)^\psi)(1 + 2\gamma \psi)}{\psi}\right) \nonumber\\
&=& \frac{\psi}{1+ 2\gamma \psi}\left(1-\alpha(2\gamma) - \frac{(1 - (1-2\gamma)^\psi)}{\psi} - 2\gamma(1 - (1-2\gamma)^\psi)\right) \nonumber\\
&\geq& \frac{\psi}{1+ 2\gamma \psi}\left(1-\alpha(2\gamma) - \frac{1 - (1 - 2\gamma\psi)}{\psi} - 2\gamma\right) \label{eq:induc_pow2}\\
&=& \frac{\psi}{1+ 2\gamma \psi}\left(1-\alpha(2\gamma) - 4 \gamma\right) \nonumber \\
&=& \frac{\psi}{1+ 2\gamma \psi}\left(1-\theoremC\sqrt{2\gamma\ln(1/(2\gamma))} - 4 \gamma\right)\nonumber \\
&\geq& \frac{\psi}{1+ 2\gamma \psi}\left(1-\theoremC\sqrt{2\gamma\ln(1/\gamma)} - 4 \gamma\right)\nonumber \\
&=& \frac{\psi}{1+ 2\gamma \psi}\left(1-\sqrt{2}\alpha(\gamma) - 4
\gamma\right). \nonumber
\end{eqnarray}
Where the Inequality (\ref{eq:induc_pow2}) follows since for any $n\in\mathrm{N}$ and any $x \in (0,1)$,  $(1-x)^n > (1-nx)$.
\qed\end{proof}

\begin{lemma}\label{lem:notSimpleArithmetics1}
For any $k\in \mathrm{N}$ and any $q \geq \ceil{k/2}$
$$\frac{\ceil{k/2}}{\tau + q - 1} \geq \frac{1}{2}\left(1 - \frac{(q - \ceil{\half{k}})}{k}\right) \ .$$
\end{lemma}
\begin{proof}
\begin{eqnarray}
\frac{\ceil{k/2}}{\tau + q - 1}
&=& \frac{\ceil{k/2}}{\ceil{k/2} + q - 1} \nonumber\\
&=& \frac{\ceil{k/2}}{2\ceil{k/2} + (q - \ceil{k/2}) - 1}  \nonumber\\
&\geq& \frac{\ceil{k/2}}{2\ceil{k/2} + (q - \ceil{k/2})}  \nonumber\\
&=& \left(\frac{1}{2} - \frac{\frac{1}{2} (q - \ceil{\half{k}})}{2\ceil{\half{k}}+(q - \ceil{\half{k}})}\right) \nonumber\\
&\geq& \left(\frac{1}{2} - \frac{1}{2}\cdot\frac{q - \ceil{\half{k}}}{2\ceil{\half{k}}}\right) \label{eq:notSimpleArithmeticsSub11} \\
&\geq& \frac{1}{2}\left(1 - \frac{q - \ceil{\half{k}}}{2(k/2)}\right) \label{eq:notSimpleArithmeticsSub12} \\
&\geq& \frac{1}{2}\left(1 - \frac{q - \ceil{\half{k}}}{k}\right) \ .\label{eq:notSimpleArithmeticsSub13}
\end{eqnarray}
where  (\ref{eq:notSimpleArithmeticsSub11}), (\ref{eq:notSimpleArithmeticsSub12}) and (\ref{eq:notSimpleArithmeticsSub13}) follow since $q \geq \ceil{k/2}$.
\qed\end{proof}

The following Lemma is used in the proof of Lemma \ref{lemma:prof_r}, for bounding the profit of the algorithm from items in the 2nd half, and conditioned on there being $q$ such items above the threshold.

\begin{lemma}\label{lem:notSimpleArithmetics2}
For any $k \in \mathrm{N}$, $\ceil{k/2} \leq q \leq 2\ceil{k/2}$ and any $\gamma$ that satisfies $k \leq 1/\gamma$ we have
\begin{equation}\label{eq:longAlgebraStatement}
\frac{1}{1+2\gamma q} \left(1 - \frac{q - \ceil{k/2}}{k}\right) \geq \frac{1}{1+\gamma k}\left(1 - 2\cdot\frac{q-k/2}{k}\right).
\end{equation}
\end{lemma}
\begin{proof}
While the factor on the left in the left hand side of (\ref{eq:longAlgebraStatement})  is smaller than factor on the left in the right hand side of (\ref{eq:longAlgebraStatement}), we show that this is compensated for the factors on the right of both sides of the Inequality.

Note that for any $\ceil{k/2} \leq q \leq 2\ceil{k/2}$,
\begin{eqnarray}
1 - \frac{q - \ceil{k/2}}{k}
&\geq& 1 - \frac{\ceil{k/2}}{k} \label{ineq:notsimple}
\\&\geq& 1 - \frac{k/2 + 1/2}{k}\nonumber
\\&=& 1/2 \cdot (1 - 1/k) \nonumber
\\&\geq& 0. \label{ineq:notsimpleagain}
\end{eqnarray}
The derivation (\ref{ineq:notsimple}) above follows by substituting the maximal value $q$ can take.
In Equation (\ref{eq:longAlgebraStatement}) both factors $1/(1+2\gamma q)$ (on the left) and $1/(1+\gamma k)$ (on the right) are strictly positive since both $q$ and $k$ are strictly positive. It therefore follows from (\ref{ineq:notsimpleagain}) that the left hand side of Equation (\ref{eq:longAlgebraStatement}) is non negative. It also follows that if $1 - 2(q-k/2)/k \leq 0$ then the right hand side of Equation (\ref{eq:longAlgebraStatement}) is  $\leq 0$ and the lemma holds.

Thus, we may assume that the right hand side of Equation (\ref{eq:longAlgebraStatement}) is strictly positive, and thus that \begin{equation}1 - 2\cdot\frac{q-k/2}{k}> 0.\label{eq:notsimplepos1}\end{equation}
As $q\geq \ceil{k/2}$ it follows that $q-k/2\geq 0$ and thus inequality (\ref{eq:notsimplepos1}) implies that \begin{equation}1 - \frac{q-k/2}{k}> 0.\label{eq:notsimplepos2}\end{equation}

As $\ceil{k/2} \geq k/2$ we can bound the left hand side of Equation (\ref{eq:longAlgebraStatement}) as follows:  \begin{equation}\frac{1}{1+2\gamma q} \left(1 - \frac{q - \ceil{k/2}}{k}\right)
\geq\frac{1}{1+2\gamma q} \left(1 - \frac{q - k/2}{k}\right).\label{eq:notsimplenoceil}\end{equation}
As the right hand side of Equation (\ref{eq:longAlgebraStatement}) is strictly positive, we can multiply and divide the right hand side of (\ref{eq:notsimplenoceil}) by the right hand side of Equation (\ref{eq:longAlgebraStatement}):
\begin{eqnarray}
\frac{1}{1+2\gamma q} \left(1 - \frac{q - k/2}{k}\right)
&=&\frac{1}{1+\gamma k}\left(1 - 2\cdot\frac{q-k/2}{k}\right) \label{eq:shit_term1}\\
   &&\qquad  \cdot\left(\frac{1+\gamma k}{1+2\gamma q}\right)
    \cdot\left(\frac{1 - \frac{q - k/2}{k}}{1 - 2\cdot\frac{q - k/2}{k}}\right). \nonumber
\end{eqnarray}
Note the the right hand side above in Equation (\ref{eq:shit_term1}) is of the form $\mathrm{RHS}(\mathrm{Equation\ }\ref{eq:longAlgebraStatement})$ times some factor $t$. By assumption $\mathrm{RHS}(\mathrm{Equation\ }\ref{eq:longAlgebraStatement})>0$ so, given Equation (\ref{eq:notsimplenoceil}) and that $t\geq 1$ then the lemma holds, so it suffices to show that $t\geq 1$:
\begin{equation}\label{eqn:preLast}
\frac{1+\gamma k}{1+2\gamma q} \cdot\frac{\left(1 - \frac{q - k/2}{k}\right)}{\left(1 - 2\cdot\frac{q - k/2}{k}\right)} \geq 1.
\end{equation}
As $\gamma$, $k$, $q$ are non negative, and given Inequalities (\ref{eq:notsimplepos1}) and (\ref{eq:notsimplepos2}), it follows that
the numerators and denumerators, in both fractions whose product is on the left hand side of Equation (\ref{eqn:preLast}), are strictly positive.

We substitute $r$ for $q - k/2$, note that $r\geq 0$, it now follows that
\begin{eqnarray*}
\frac{1+\gamma k}{1+2\gamma q} \cdot\frac{\left(1 - \frac{q - k/2}{k}\right)}{\left(1 - 2\cdot\frac{q - k/2}{k}\right)}
&=& \frac{1+\gamma k}{1+2\gamma (k/2 + r)} \cdot\frac{\left(1 - \frac{r}{k}\right)}{\left(1 - \frac{2r}{k}\right)} \\
&=& \frac{1+\gamma k}{1+\gamma k + 2\gamma r} \cdot\frac{\left(k - r\right)}{\left(k - 2r\right)}.
\end{eqnarray*}
So it suffices to show:
$$(1+\gamma k)(k-r) \geq (1+\gamma k + 2\gamma r)(k-2r)$$
or, equivalently,
$$(1+\gamma k)(k-r) - (1+\gamma k + 2\gamma r)(k-2r)\geq 0$$
\begin{eqnarray*}
\lefteqn{(1+\gamma k)(k-r) - (1+\gamma k + 2\gamma r)(k-2r)} \\
&=&(1+\gamma k)k-(1+\gamma k)r - (1+\gamma k)k + 2(1+\gamma k)r - 2\gamma r k + 4\gamma r^2\\
&=&(1+\gamma k)r - 2\gamma r\cdot k + 4\gamma r^2\\
&=&r\bigr(1+\gamma k - 2\gamma k + 4\gamma r\bigr)\\
&=&r\bigr(1-\gamma k + 4\gamma r\bigr)\\
&\geq& 0
\end{eqnarray*}
where the last inequality holds since $k \leq 1/\gamma$ (and therefore $k\gamma \leq 1$) and $4\gamma r \geq 0$.
\qed\end{proof}


\section{Proof of Theorem \ref{THEOREM:HALF}} \label{sec:half_details}

Define \begin{eqnarray*} \alpha(x) &=&
\theoremC\sqrt{x\ln(1/x)},\\ t &=& \max\set{x \mid \alpha(x)\leq 1 }, \qquad \theoremT = \theoremTVal \approx 1/\theoremInvTValApprox.\end{eqnarray*}

\halfThm*

We first give an overview of the proof of Theorem \ref{THEOREM:HALF}, and then give the full proof in detail.

\subsection{Proof Overview}

We need the following definitions for the proof of Theorem \ref{THEOREM:HALF}. Let $Z=\set{z_1,\ldots,z_n}$ and
 define the random variable $X_i$, $1 \leq i \leq n$ to be the $i$'th smallest point in $Z$. Define the random variable
$C_i$ to be the number of $X_j$'s that lie in the interval
$[X_i,X_i+ \gamma)$. At most one of the $X_j\in[X_i,X_i+ \gamma)$
can belong to a $\gamma$-independent set.

Given any set $S$ of points we define the following \emph{greedy algorithm} that constructs a $\gamma$-independent set.
The greedy algorithm initializes the set with the smallest point in $S$ and then traverses the remaining points in increasing order and adds the point $x$ to the $\gamma$-independent set if $x$ is larger by at least $\gamma$ from any point previously added. We denote by $G(S,\gamma)$ the $\gamma$-independent set computed by applying greedy algorithm to $S$.
Let $I_i$ be a random variable with binary values where $I_i=1$ iff $X_i \in G(Z,d)$, and define $p_i = \prob{I_i=1}.$

Lemma \ref{lem:greedy_best} proves the well known fact that the greedy algorithm picks a $\gamma$-independent set of largest size. So $m(Z,\gamma) = \size{G(Z,\gamma)}$. Our proof uses this fact.

We give an outline of the proof of Theorem \ref{THEOREM:HALF}. The details are in Section \ref{sec:half_details}.
The proof splits into two main subcases:

{\bf ``Small $n$":}
In this case
the expected number of non-overlapping intervals of the form $(z_i-\gamma,z_i]$ that contain no point $z_j$, $j\neq i$ is
sufficiently large.
  Specifically, Lemma  \ref{lem:prfot_small_n} shows that for $n\leq 2\sqrt{e\cdot \ln(1/\gamma)/\gamma}$ the expected number of intervals of this form is at least $\frac{n}{1+n\gamma}(1-\alpha(\gamma))$. All such $z_i$ will be chosen by the greedy algorithm thus proving the theorem for small $n$.

{\bf ``Larger $n$":}
  For $n > 2\sqrt{e \cdot \ln(1/\gamma)/\gamma} $, we need to
  consider the expected number of points $z_j$ ``discarded" when the greedy algorithm chooses some
  $z_i$,
so the key point in this part of the proof is to bound $\Exp{C_i\mid
  I_i=1}$.

Given a set of points $Z$, since the greedy algorithm picks a maximal
$\gamma$-independent set we have that
$m(Z,\gamma)= \sum_{i=1}^n I_i$. Therefore $\Exp{m(Z,\gamma)}=\sum_{i=1}^n p_i$. So to prove the
theorem we seek a lower bound for $\sum_{i=1}^n p_i$.

We note that $\sum_{i=1}^n C_i \cdot I_i = n$. This follows because
(a) the $C_i$ for which $I_i=1$ are non-overlapping so this sum is
$\leq n$, and (b) $\sum_{i=1}^n C_i \cdot I_i < n$ implies that the
greedy algorithm skipped over a point that should have been chosen.
In particular, the expectation of $\sum_{i=1}^n C_i \cdot I_i$ is
also $n$.

Contrawise, it must be that
 \begin{equation*}E\left(\sum_{i=1}^n
C_i \cdot
I_i\right) = \sum_{i=1}^n \Exp{C_i\mid I_i=1} \cdot \mbox{\rm Prob}(I_i=1)
=\sum_{i=1}^n p_i \Exp{C_i\mid I_i=1}.
\end{equation*}

If we could show that
\begin{equation} \label{eqn:beta}
\Exp{C_i,I_i=1}\leq \beta
\end{equation}
 for all $i\in \set{1, \ldots, n}$ then it
would follow that $\sum_{i=1}^n p_i \geq n/\beta$, giving us a lower
bound on the sum of the $p_i$'s.

In fact we show (Lemma \ref{lem:cond_Ci_bound}) that the bound
(\ref{eqn:beta}) holds for $i=1,\ldots, \floor{n/4}$ with
\begin{equation} \label{eqn:beta1}
\beta = \frac{1}{n} ~+~ 1 + \frac{\gamma n}{1 -
\frac{4}{3}\sqrt{\frac{1}{2}\ln{\frac{1}{1-\gamma}} +
\frac{\ln{n}}{n}}} \ .
\end{equation}

To get an intuition of why the upper bound
\begin{equation}
\Exp{C_i,I_i=1}\leq \frac{1}{n} ~+~ 1 + \frac{\gamma n}{1 -
\frac{4}{3}\sqrt{\frac{1}{2}\ln{\frac{1}{1-\gamma}} +
\frac{\ln{n}}{n}}}  \label{eq:outline_c_i}
\end{equation}
holds for $i \le \floor{n/4}$, we note that for small $\gamma$ it is
close to $1+\gamma n$ and we expect to see $\le \gamma n$ points in
an interval of length $\gamma$ following any specific point.
 Specifically, note that when
$\gamma \rightarrow 0$ since $n > 2\sqrt{e \cdot
\ln(1/\gamma)/\gamma} $ then $n \rightarrow \infty$ and
$$\lim_{\substack{{\gamma\rightarrow 0} \\
{n > \sqrt{\frac{e
\ln(1/\gamma)}{\gamma}}}}}{\frac{4}{3}\sqrt{\frac{1}{2}\ln{\frac{1}{1-\gamma}}
+ \frac{\ln{n}}{n}}} = 0 \ ,$$ so  $\Exp{C_i \mid I_i = 1} \approx
{1+\gamma n}.$

We prove Inequality (\ref{eq:outline_c_i})  for $i=1,\ldots,
\floor{n/4}$  in two steps. Lemma \ref{lem:conditional_x_i_bound}
shows that if the greedy algorithm chooses $q_i$  then
 with high probability
 \begin{equation} \label{eqn:mean}
X_i \leq \frac{i-1}{n}+ \sqrt{\frac{1}{2}\ln{\frac{1}{1-\gamma}} +
\frac{\ln{n}}{n}}  \ .
\end{equation}
 (Informally this is to say that $X_i$ is
concentrated around its mean, which is $i/(n+1)$). Next, Lemma
\ref{lem:independent} and Lemma \ref{lem:ci_cond_xi} imply that if
the greedy algorithm chooses $X_i$ then
\begin{equation} \label{eqn:mm}
\Exp{C_i \mid X_i = x_i}= 1 + \frac{1}{1-x_i}(n - i).
\end{equation}
Substituting Equation (\ref{eqn:mean}) in (\ref{eqn:mm}) we get
(\ref{eq:outline_c_i}).

The claim that $n = \sum_{i=0}^{n}{p_i \Exp{C_i \mid I_i = 1}}$
which we mentioned before, extends (with the same argument) to
partial sums. {\sl I.e.}, for all $k \le n$
$$k \leq \sum_{i=0}^{k}{p_i \Exp{C_i \mid I_i = 1}} \ .$$
Using this with $k = \floor{n/4}$ and  substituting the bound
(\ref{eqn:beta}) which we have proven for $1 \leq i \leq
\floor{n/4}$, we get that $\sum_{i=0}^{\floor{n/4}}{p_i} \geq
{1/\beta\cdot\floor{n/4}}$. Let $A_1=\set{X_1, \ldots, X_{\lfloor n/4 \rfloor}}$, we have that the expected size of the maximal independent set in $A$ is at least $\floor{n/4}/\beta$.

In Lemma \ref{lem:iF_eq} we show that for any $0 \leq i < 4$ the expected size of the maximal $\gamma$-independent set in the point sets $A_i=\set{X_{i\floor{n/4}+1}, \ldots, X_{(i+1) \floor{n/4 }}}$ are equal (and thus at least $\floor{n/4}/\beta$).

It therefore follows that in the union $A_1 \cup A_2 \cup A_3 \cup A_4$ the size of the maximal $\gamma$-independent set is at least $4\floor{n/4}/\beta-3$. Ergo, the size of the maximal $\gamma$-independent set in $q_1, \ldots, q_n$ is at least $(n-1)/\beta -3\geq n/\beta - 4$.
 As the greedy algorithm is optimal (Lemma \ref{lem:greedy_best}) it follows that the greedy algorithm gives a $\gamma$-independent set of size at least $n/\beta - 4$.

The theorem follows by showing that for ``large" n,
$$\frac{n}{\beta} - O(1) \geq \frac{n}{1+n \gamma}(1-\alpha(\gamma)).$$

\subsection{Full Proof of theorem \ref{THEOREM:HALF}}

\begin{proof} First let make the following definitions:
\begin{itemize}
    \item Define $Z$ be the set $\set{z_1, z_2, \ldots, z_n}$ where each $z_i$ is sampled uniformly in $[0,1)$.
    \item Define the random variable $X_i$, $1 \leq i \leq n$ to be the $i$'th smallest point in $Z$.
    \item Define the random variable $C_i$ to be the number of $X_j$'s that lie in the interval $[X_i,X_i+ \gamma)$. Note that at most one of the $X_j\in[X_i,X_i+ \gamma)$ can belong to a $\gamma$-independent set.
    \item Given any set $S$ of points we define the following \emph{greedy algorithm} that constructs a $\gamma$-independent set.
The greedy algorithm initializes the set with the smallest point in $S$ and then traverses the remaining points in increasing order and adds the point $x$ to the $\gamma$-independent set if $x$ is larger by at least $\gamma$ from any point previously added. We denote by $G(S,\gamma)$ the $\gamma$-independent set computed by applying greedy algorithm to $S$.
Note that $\size{G(Z, \gamma)} = m(Z, \gamma)$ (Lemma (\ref{lem:greedy_best})).
    \item Let $G(Z,\gamma)$ be a random variable which is the $\gamma$-independent set obtained by the greedy algorithm when applied to the set $Z$.
    \item Let $I_i$ be a random variable with binary values where $I_i=1$ iff $X_i \in G(Z,d)$, and define $$p_i = \prob{I_i=1}.$$
\end{itemize}

For  $n \leq
2\sqrt{e \ln{(1/\gamma) \cdot 1/\gamma}}$ the statement follows from Lemma \ref{lem:prfot_small_n}.

So from here to the end of the proof we  assume that $n >  2\sqrt{e \ln{(1/\gamma) \cdot 1/\gamma}}$ .

Let $A_{i}$, $1\leq i < \theoremFVal$, be the set $\set{X_{1 + i\floor{n/\theoremFVal}}, \ldots, X_{(i+1)\floor{n/\theoremFVal}}}$.

Note that for any  $1\leq i < \theoremFVal$
$$\max{\left(A_{i}\right)}=
X_{(i+1)\floor{\frac{n}{\theoremFVal}}} \leq
X_{1+(i+1)\floor{\frac{n}{\theoremFVal}}} =
\min{\left(A_{i+1}\right)}  \ .$$

It is easy to verify that
$$
\Exp{m(Z, \gamma)} \geq \Exp{m\left(\bigcup_{i = 0}^{\theoremFMinOneVal} {A_{i}}, \gamma\right)}
\geq \sum_{i = 0}^{\theoremFMinOneVal}\Exp{m(A_{i}, \gamma)}-\theoremFMinOneVal
$$
By Lemma \ref{lem:iF_eq}, $\Exp{m(A_{i}, \gamma)} = \Exp{m(A_{0}, \gamma)}$ for $i=1,2,3$. By Lemma \ref{lem:greedy_best},
$m(Z,\gamma)$ is the size of the independent set picked by the greedy algorithm and
 by the definition of the $p_i$'s we have that $\Exp{m(A_{0}, \gamma)} = \sum_{i=1}^{\floor{n/\theoremFVal}}{p_i}$. So we get
\begin{equation} \label{eq:sumpi} \Exp{m(Z, \gamma)} \geq \theoremFVal\cdot\sum_{i=1}^{\floor{n/\theoremFVal}}{p_i}-\theoremFMinOneVal \ .
\end{equation}

For simplicity lets $r(\gamma)$ be as define is Lemma \ref{lem:half_expect}
$$r(\gamma) = \theoremK\left(\sqrt{\frac{1}{2}\ln{\frac{1}{1-\gamma}} + \frac{\beta(\gamma)}{\theoremFVal}(\gamma\ln{(1/\gamma)}) + \frac{\beta(\gamma)\ln{(\beta(\gamma))}}{\theoremFVal}\gamma}\right)\ , $$
where $\theoremK \approx \theoremKVal$ is as defined in Lemma \ref{lem:half_expect} and $\beta(\gamma) = 1+ \frac{2}{\ln{(1/\gamma)}-1}$ is as defined in Lemma \ref{lem:logngamman}.

Substituting the lower bound on  $\sum_{i=1}^{\floor{n/\theoremFVal}}{p_i}$ given by Lemma \ref{lem:half_expect} in Equation (\ref{eq:sumpi}) we get

\begin{eqnarray*}
\Exp{m(Z, \gamma)}
&\geq & \theoremFVal\cdot\frac{n}{\theoremFVal} \cdot \frac{1}{1+\gamma n}\left(1 - \theoremRGamma - \frac{\theoremFVal}{n}\right) - \theoremFMinOneVal\\
&=& \frac{n}{1+\gamma n}\left(1 - \theoremRGamma - \frac{\theoremFVal}{n} - \theoremFMinOneVal\cdot\frac{1 + \gamma n}{n}\right) \\
&=& \frac{n}{1+\gamma n}\left(1 - \theoremRGamma - \frac{\theoremFVal}{n} - \theoremFMinOneVal\left(\frac{1}{n} + \gamma\right)\right) \\
&=& \frac{n}{1+\gamma n}\left(1 - \left(\theoremRGamma + \frac{\theoremTwoFMinOneVal}{n} + \theoremFMinOneVal\gamma\right)\right)
\end{eqnarray*}

Since $n > 2\sqrt{e\ln{(1/\gamma)}\cdot1/\gamma}$ then
\begin{eqnarray}
\Exp{m(Z, \gamma)}
&\geq& \frac{n}{1+\gamma n}\left(1 - \left(\theoremRGamma + \frac{\theoremTwoFMinOneVal}{n} + \theoremFMinOneVal\gamma\right)\right) \nonumber \\
&\geq& \frac{n}{1+\gamma n}\left(1 - \left(\theoremRGamma + \frac{\theoremTwoFMinOneVal}{2\sqrt{e}}\cdot\frac{\sqrt{\gamma}}{\sqrt{\ln{(1/\gamma)}}} + \theoremFMinOneVal\gamma\right)\right) \label{eq:maineq}
\end{eqnarray}

Now we upper bound the terms $\theoremRGamma$, $\frac{\theoremTwoFMinOneVal}{2\sqrt{e}}\frac{\sqrt{\gamma}}{\sqrt{\ln{(1/\gamma)}}} $, and $\theoremFMinOneVal\gamma$ for $0< \gamma < \theoremT$.
We  start with $\theoremRGamma$.
\begin{eqnarray*}
\theoremRGamma &=& \theoremK\sqrt{\frac{1}{2}\ln{\frac{1}{1-\gamma}} + \frac{\beta(\gamma)}{\theoremFVal}(\gamma\ln{(1/\gamma)}) + \frac{\beta(\gamma)\ln{(\beta(\gamma))}}{\theoremFVal}\gamma} \\
&=& \sqrt{\gamma\ln{(1/\gamma)}}\cdot \theoremK\sqrt{\frac{1}{2}\cdot\frac{\log(1/(1-\gamma))}{\gamma\ln{(1/\gamma)}} + \frac{\beta(\gamma)}{\theoremFVal} + \frac{\beta(\gamma)\ln{(\beta(\gamma))}}{\theoremFVal\ln{(1/\gamma)}}} \\
&=& \sqrt{\gamma\ln{(1/\gamma)}}\cdot
\theoremK\sqrt{\frac{1}{2}\cdot\frac{\log(1/(1-\gamma)^{1/\gamma})}{\ln{(1/\gamma)}}
+ \frac{\beta(\gamma)}{\theoremFVal} + \frac{\beta(\gamma)\ln{(\beta(\gamma))}}{\theoremFVal\ln{(1/\gamma)}}} \\
& \le &
\sqrt{\gamma\ln{(1/\gamma)}}\cdot \theoremK\sqrt{\frac{1}{2}\cdot\frac{\log(1/(1-\theoremT)^{1/\theoremT})}{\ln{(1/\theoremT)}} + \frac{\beta(t)}{\theoremFVal} + \frac{\beta(t)\ln{(\beta(t))}}{\theoremFVal\ln{(1/\theoremT)}}}  \ ,
\end{eqnarray*}
where the last inequality follows since  the functions
 $\frac{1}{\ln{(1/\gamma)}}$,
$\ln{(1/(1-\gamma)^{1/\gamma})}$ and $\beta(\gamma)$ are monotonic increasing in $\gamma$. So we have that
$$
\theoremRGamma \le \theoremCOne \sqrt{\gamma\ln{(1/\gamma)}}$$ where
 $\theoremCOne =
 \theoremK\sqrt{\frac{1}{2}\cdot\frac{\log(1/(1-\theoremT)^{1/\theoremT})}{\ln{(1/\theoremT)}}
    + \frac{\beta(t)}{\theoremFVal} + \frac{\beta(t)\ln{(\beta(t))}}{\theoremFVal\ln{(1/\theoremT)}}}
 \approx \theoremCOneVal
 $
 is a constant.

Now we look at the  term $\frac{\theoremTwoFMinOneVal}{2\sqrt{e}}\cdot\frac{\sqrt{\gamma}}{\sqrt{\ln{(1/\gamma)}}}$:
\begin{eqnarray*}
\frac{\theoremTwoFMinOneVal}{2\sqrt{e}}\cdot\frac{\sqrt{\gamma}}{\sqrt{\ln{(1/\gamma)}}}
&=& \frac{\theoremTwoFMinOneVal}{2\sqrt{e}}\frac{1}{\ln{(1/\gamma)}} \sqrt{\gamma\ln{(1/\gamma)}} \\
&\leq&
\frac{\theoremTwoFMinOneVal}{2\sqrt{e}\ln{(1/\theoremT)}} \cdot \sqrt{\gamma\ln{(1/\gamma)}} = \theoremCTwo \sqrt{\gamma\ln{(1/\gamma)}}\ ,
\end{eqnarray*}
where the last inequality follows since
the function
 $\frac{1}{\ln{(1/\gamma)}}$,
is monotonic increasing in $\gamma$, and $\theoremCTwo = \frac{\theoremTwoFMinOneVal}{2\sqrt{e}\ln{(1/\theoremT)}} \approx \theoremCTwoVal$.

Similarly for the   term $\theoremFMinOneVal\gamma$
\begin{eqnarray*}
\theoremFMinOneVal\gamma & = & \theoremFMinOneVal\sqrt{\frac{\gamma}{\ln{(1/\gamma)}}}\sqrt{\gamma\ln{(1/\gamma)}} \\
& \le & \theoremFMinOneVal\sqrt{\frac{t}{\ln{(1/t)}}}\sqrt{\gamma\ln{(1/\gamma)}} = \theoremCThree \sqrt{\gamma\ln{(1/\gamma)}} \ ,
\end{eqnarray*}
where $\theoremCThree =\theoremFMinOneVal\sqrt{\frac{t}{\ln{(1/t)}}} \approx \theoremCThreeVal$.

Summing  all three term together we get that
$$
\theoremRGamma  + \frac{\theoremTwoFMinOneVal}{2\sqrt{e}}\cdot\frac{\sqrt{\gamma}}{\sqrt{\ln{(1/\gamma)}}} + \theoremFMinOneVal\gamma \le (\theoremCOne + \theoremCTwo + \theoremCThree)\sqrt{\gamma\ln{(1/\gamma)}} \le \theoremC \sqrt{\gamma\ln{(1/\gamma)}}\ .
$$
Substituting this bound back in Equation (\ref{eq:maineq}) completes the proof.
\qed\end{proof}

The following lemma is well known and can be proved by induction on $\size{S}$.

\begin{lemma}\label{lem:greedy_best}
For any set $S$, $\size{G(S,\gamma)} = m(S,\gamma)$.
\end{lemma}

The following lemma shows that when  ``$n$ is small" as a function
of $\gamma$, Equation (\ref{eq:theorem}) holds; thus proving Theorem
\ref{THEOREM:HALF} for this case. In the proof of this lemma we will use
some arithmetical lemmas which are proved later (Lemmas \ref{len:1_min_1_n} through \ref{lem:1_min_gamma})

\begin{lemma}\label{lem:prfot_small_n}
For any $0 < \gamma \leq \theoremT$, if $n \leq
2\sqrt{e\ln(1/\gamma)\cdot1/\gamma}$, then:
$$\Exp{m(Z,\gamma)} \geq \frac{n}{1+\gamma n} \left(1-\theoremC \sqrt{\gamma\ln{(1/\gamma)}}\right) \ .$$
\end{lemma}
\begin{proof}
First we prove the lemma for $n=1$ (Since we would like to use Lemma \ref{lem:1_min_gamma} that works only for $n \ge 2$).
For $n=1$, $x_1 = z_1$ is $\gamma$-independent so $\Exp{m(Z,\gamma)} = 1$ and the right hand side of Equation in the statement is at most $1$.

We now prove the lemma for $n \geq 2$.

For  $n \geq 2$, using Lemma \ref{lem:min_exp}, we derive that
\begin{eqnarray}
 \Exp{m(Z,\gamma)} &\geq& n(1-\gamma)^{n-1} \nonumber \\
 &\geq & n(1-\gamma)^n \nonumber\\
 &=& n\left((1-\gamma)^{1/\gamma}\right)^ {\gamma n} \label{eq:lem42}
\\&\geq &n \left((1-\gamma)/e\right)^{\gamma n} \label{eq:lem43}
\\&=&n (1 - \gamma) ^ {\gamma n} (1/e)^{\gamma n} \label{eq:lem44} \ .
\end{eqnarray}
The expression in
 (\ref{eq:lem42}) is no smaller than the expression in (\ref{eq:lem43}) since $(1-\gamma)^{1/\gamma-1} \ge 1/e$ for any $0 < \gamma < 1$, see Lemma \ref{len:1_min_1_n}.
We now give a lower bound for the terms in Equation (\ref{eq:lem44}). For $(1-\gamma)^{\gamma n}$ we have
\begin{eqnarray}
(1-\gamma)^{\gamma n} &\geq&  (1-\gamma)^{2\sqrt{e\gamma\ln{(1/\gamma)}}} \label{eq:lem45}  \\
&\geq& (1-\gamma)^2 \label{eq:lem46} \\
&\geq& \frac{1}{1+\gamma n}(1-\gamma) \label{eq:lem47} \ .
\end{eqnarray}
Inequality (\ref{eq:lem45}) follows since $n \leq 2\sqrt{e\ln{(1/\gamma)}\cdot1/\gamma}$. Inequality (\ref{eq:lem46}) follows since the function $ x\ln(1/x)$ is smaller than $1/e$ for every $0<x$, see Lemma \ref{lem:arithmetic_xlnx}. Inequality (\ref{eq:lem47}) follows since $(1-\gamma) \ge 1/(1+n\gamma)$ for $n\ge 2$ and $0<\gamma < 1/2$, see Lemma \ref{lem:1_min_gamma}.

For $(1/e)^{\gamma n}$, using the
 Taylor expansion of $e^{-x}$ around $0$ we get that for any $x$ there exists $\theta$ s.t. $e^{-x} = 1 - x + (e^{-\theta}) \frac{x^2}{2} \geq 1-x$, so
$$(1/e)^{\gamma n} \geq (1-\gamma n)\ .$$

Plugging these lower bounds back into Equation (\ref{eq:lem44}) we get
\begin{eqnarray}
 \Exp{m(Z,\gamma)} &\geq& n (1 - \gamma) ^ {\gamma n} (1/e)^{\gamma n} \nonumber  \\
 &\geq& n\frac{1}{1+\gamma n}(1 - \gamma)(1-\gamma n)  \nonumber \\
 &=& \frac{n}{1+\gamma n}\left( 1 - \gamma n - \gamma + \gamma^2 n\right) \nonumber \\
 &\geq& \frac{n}{1+\gamma n}\left( 1 - \gamma n - \gamma\right) \nonumber \\
 &\geq& \frac{n}{1+\gamma n}\left( 1 - \gamma \cdot2\sqrt{e\ln{(1/\gamma)}\cdot{1/\gamma}} - \gamma\right) \label{eq:lem05_03} \\
 &=& \frac{n}{1+\gamma n}\left( 1 - 2\sqrt{e}\cdot\sqrt{\gamma\ln{(1/\gamma)}} - \gamma\right) \nonumber \\
 &=& \frac{n}{1+\gamma n}\left( 1 - \sqrt{\gamma\ln{(1/\gamma)}}\left(2\sqrt{e} + \sqrt{\frac{\gamma}{\ln{(1/\gamma)}}}\right)\right) \nonumber \\
&\geq&
 \frac{n}{1+\gamma n}\left( 1 - \sqrt{\gamma\ln{(1/\gamma)}}\left(2\sqrt{e} + \sqrt{\frac{\theoremT}{\ln{(1/\theoremT)}}}\right)\right) \label{eq:lem05_04} \\
&\geq & \frac{n}{1+n\gamma}\left(1-\theoremC \sqrt{\gamma\ln{(1/\gamma)}}\right) \ . \label{eq:lem05_05}
\end{eqnarray}
 Inequality (\ref{eq:lem05_03}) follows since $n \leq 2\sqrt{e\ln{(1/\gamma)}\cdot1/\gamma}$.
Inequality (\ref{eq:lem05_04}) follows since the function
$\frac{1}{\ln{(1/\gamma)}}$ is monotonically increasing with $\gamma$. Inequality (\ref{eq:lem05_05}) follows by substituting the value of $\theoremT$.
\qed\end{proof}

\begin{lemma}\label{lem:min_pi}
$\prob{I_i = 1} \geq (1-\gamma)^{n-1}$.
\end{lemma}
\begin{proof}
Note that for all $1 \leq i,k \leq n$, the probability that the $k$th sample of $Z$  is the $i$th smallest (i.e.\ the probability that $X_i=y_k$) is exactly $1/n$. Also, note that for any $1\leq k,j \leq n$, $k\neq j$, the probability that $y_j$ falls into the real interval $[y_k-\gamma,y_k)$ is at most  $\gamma$. It therefore follows that
\begin{eqnarray*} \prob{I_i = 1} &\geq& \sum_{k=1}^{n}{\prob{X_i = y_k}\cdot \mbox{Prob}\bigr[\mbox{for all\ }j\neq k, y_j \not\in [y_k - \gamma, y_k)\bigr]} \\
&=&
\frac{1}{n}\sum_{k=1}^{n}(1-\gamma)^{n-1}
=
\frac{1}{n}\cdot n (1-\gamma)^{n-1} = (1-\gamma)^{n-1}  \ .
\end{eqnarray*}
\qed\end{proof}

\begin{lemma}\label{lem:min_exp}
$\Exp{m(Z,\gamma)} \geq n(1-\gamma)^{n-1}$
\end{lemma}
\begin{proof}
We have that $\Exp{m(Z,\gamma)} = \Exp{\sum_{i=1}^{n}{I_i}} =
\sum_{i=1}^{n}\prob{I_i = 1} \geq n(1-\gamma)^{n-1}$, by Lemma
\ref{lem:min_pi}.
\qed\end{proof}

\begin{lemma}\label{lem:half_expect}
For any $\gamma \leq \theoremT$, $\theoremT=\theoremTVal$, and any $n \geq
2\sqrt{e\ln{(1/\gamma)}\cdot 1/\gamma}$ we have that
$$
\sum_{i=1}^{\floor{n/\theoremFVal}}{p_i}
\geq \frac{n}{\theoremFVal} \cdot \frac{1}{1+\gamma n}\left(1 -
r(\gamma) - \frac{\theoremFVal}{n}\right)
$$
where
$$r(\gamma) = \theoremK\left(\sqrt{\frac{1}{2}\ln{\frac{1}{1-\gamma}} + \frac{\beta(\gamma)}{\theoremFVal}(\gamma\ln{(1/\gamma)}) + \frac{\beta(\gamma)\ln{(\beta(\gamma))}}{\theoremFVal}\gamma}\right) $$
and
$\theoremK \approx \theoremKVal$,
is defined precisely in the proof below and $\beta(\gamma) = 1+ \frac{2}{\ln{(1/\gamma)}-1}$ is define in Lemma \ref{lem:logngamman}.
\end{lemma}
\begin{proof}
From Lemma \ref{lem:Ci_sum} it follows that  $\floor{\frac{n}{\theoremFVal}}
\leq \sum_{i=1}^{\floor{\frac{n}{\theoremFVal}}}{p_i}\Exp{C_i \mid I_i=1}$. Using
 Lemma \ref{lem:cond_Ci_bound} to upper bound  $\Exp{C_i \mid I_i}$ we get that:
\begin{eqnarray}
\floor{\frac{n}{\theoremFVal}} &\leq&
\sum_{i=1}^{\floor{n / \theoremFVal}}{p_i}\left(\frac{1}{n} ~+~
  1+\frac{\gamma n}{1 - \theoremFOverFVal\sqrt{ \frac{1}{2}\ln{\frac{1}{1-\gamma}}+\frac{\ln{n}}{n}}} \right) \nonumber \\
&\leq& \frac{1}{n}\sum_{1}^{\floor{n / \theoremFVal}}{p_i} ~+~
  \left(1+\frac{\gamma n}{1 - \theoremFOverFVal\sqrt{ \frac{1}{2}\ln{\frac{1}{1-\gamma}}+\frac{\ln{n}}{n}}} \right)\sum_{1}^{\floor{n / \theoremFVal}}{p_i} \nonumber\\
&\leq& \frac{1}{\theoremFVal} ~+~
  \left(1+\frac{\gamma n}{1 - g(n,\gamma)} \right)\sum_{1}^{\floor{n / \theoremFVal}}{p_i} \label{eq:p_i_right}
\end{eqnarray}
where
$$g(n,\gamma) =\theoremFOverFVal \sqrt{\frac{1}{2}\ln{\frac{1}{1-\gamma}}+\frac{\ln{n}}{n}}$$

Now we give an upper bound for the factor $1+\frac{\gamma n}{1 - g(n,\gamma)}$ in Equation (\ref{eq:p_i_right}).

The function $g(n,\gamma)$ is increasing with $\gamma$ for $\gamma \in [0,t)$ and decreasing  with $n$ for $n > e$ (use Lemma \ref{lem:arithmetic_xlnx} with $x = 1/n$). Therefore for $n \geq
2\sqrt{e\ln{(1/\gamma)}\cdot 1/\gamma}$ and for $\gamma < t$ we have that $g\left(n,\gamma\right) \le g\left(2\sqrt{\frac{e\ln{(1/\theoremT)}}{\theoremT}}, \theoremT\right) \equiv h \approx \theoremHVal$.
Since for any $x\neq 1$ we have that $\frac{1}{1-x} = 1 + x + \frac{x^2}{1-x}$ and since $g(n, \gamma) \neq 1$ we get that:
\begin{eqnarray}
1+\frac{\gamma n}{1 - g(n,\gamma)} &=& {1 + \gamma n\left(1 + g(n,\gamma) + g(n,\gamma)^2 \frac{1}{1 - g(n, \gamma)}\right)} \nonumber \\
&=& (1 + \gamma n) + \gamma n \cdot g(n,\gamma)\left(1 + \frac{g(n,\gamma) }{1 - g(n, \gamma)}\right) \nonumber \\
&\leq& (1 + \gamma n) + \gamma n \cdot g(n,\gamma)\left(1 + \frac{h}{1 - h}\right) \nonumber \\
&=& (1 + \gamma n)\left(1 + \frac{\gamma n}{1 + \gamma n} g(n,\gamma)\left(1 + \frac{h}{1 - h}\right)\right) \label{eq:first_mul_8} \ .
\end{eqnarray}

Next, we give an upper bound on $\left(1 + \frac{h}{1 - h}\right)\frac{\gamma n}{1+\gamma n}  g(n,\gamma)$.
\begin{eqnarray}
\lefteqn{\left(1 + \frac{h}{1 - h}\right)\frac{\gamma n}{1+\gamma n} g(n,\gamma)} \nonumber \\
&&= \left(1 + \frac{h}{1 - h}\right)\frac{\gamma n}{1+\gamma n} \theoremFOverFVal\sqrt{\frac{1}{2}\ln{\frac{1}{1-\gamma}} + \frac{\ln{n}}{n}} \nonumber \\
&&= \theoremK \sqrt{\frac{1}{2}\ln{\frac{1}{1-\gamma}}\left(\frac{\gamma n}{1+\gamma n}\right)^{2} + \frac{\ln{n}}{n}\left(\frac{\gamma n}{1+\gamma n}\right)^{2}} \nonumber \\
&&\leq \theoremK\sqrt{\frac{1}{2}\ln{\frac{1}{1-\gamma}} + \frac{\ln{n}}{n}\left(\frac{\gamma n}{1+\gamma n}\right)^{2}} \nonumber \\
&&= \theoremK\sqrt{\frac{1}{2}\ln{\frac{1}{1-\gamma}} + \frac{n \ln n}{(n+1/\gamma)^2}} \nonumber \\
&&\leq \theoremK\sqrt{\frac{1}{2}\ln{\frac{1}{1-\gamma}} +
\frac{\beta(\gamma)}{\theoremFVal}(\gamma\ln{(1/\gamma)}) +
\frac{\beta(\gamma)\ln{(\beta(\gamma))}}{\theoremFVal}\gamma} \label{eq:long_lemma_3}\\
&&= r(\gamma) \label{eq:long_lemma_rgamma}
\end{eqnarray}
where $k \equiv \theoremFOverFVal\left(1 + \frac{h}{1 - h}\right)$ and $r(\gamma)$ is as defined in the statement of the lemma above. Inequality (\ref{eq:long_lemma_3}) follows from Lemma \ref{lem:logngamman}.

Substituting the bound from Equation (\ref{eq:long_lemma_rgamma}) into Equation (\ref{eq:first_mul_8}) and the bound from Equation (\ref{eq:first_mul_8}) into Equation (\ref{eq:p_i_right}) we get
$$
\floor{\frac{n}{\theoremFVal}} \leq \frac{1}{\theoremFVal} + (1+\gamma n)(1 + r(\gamma))\sum_{1}^{\floor{n / \theoremFVal}}{p_i}
$$

Isolating $\sum{p_i}$ we obtain
\begin{eqnarray}
\sum_{1}^{\floor{n / \theoremFVal}}{p_i}
&\geq& \left(\floor{\frac{n}{\theoremFVal}} - \frac{1}{\theoremFVal}\right) \frac{1}{(1+\gamma n)(1 + r(\gamma))} \nonumber\\
&\geq& \left(\floor{\frac{n}{\theoremFVal}} - \frac{1}{\theoremFVal}\right) \frac{1}{(1+\gamma n)}(1 - r(\gamma)) \label{eq:one_minus_x} \\
&\geq& \left(\frac{n}{\theoremFVal} - \theoremFOverOverFVal - \frac{1}{\theoremFVal}\right)\frac{1}{1 + \gamma n}(1 - r(\gamma)) \nonumber \\
&=& \left(\frac{n}{\theoremFVal}-1\right)\frac{1}{1+\gamma n}(1 - r(\gamma)) \nonumber\\
&\geq& \frac{n}{\theoremFVal}(1 - r(\gamma)) - \frac{1}{1 + \gamma n} \nonumber\\
&=& \frac{n}{\theoremFVal}\cdot\frac{1}{1+\gamma n}\left(1 - r(\gamma) - \frac{\theoremFVal}{n}\right) \nonumber
\end{eqnarray}
where Inequality (\ref{eq:one_minus_x}) follows since $1/(1+x) \leq 1-x$ for $x \geq -1 $.
\qed\end{proof}

\begin{lemma}\label{lem:Ci_sum}
For all $1\leq k \leq n$, $\sum_{i=1}^{k}{p_i}\Exp{C_i \mid I_i=1}
\geq k$.
\end{lemma}
\begin{proof}
Recall that $X_j$ is a random variable equal to the $j$th point from the left. For a given instantiation of the $X_j$ let $Q_k$ be the subset of the $k$ leftmost points that were chosen by the greedy algorithm. Note that $i\in Q_k$ iff the instantiation of $I_i$ is one.

The union of the intervals $[X_j,X_j+\gamma)$, $j\in Q_k$, must include the instantiations of $X_1, \ldots, X_k$.
This follows since otherwise there would be some point that was not selected by the greedy choice and that could be selected, contradicting the definition of the greedy algorithm. Hence:
$\sum_{i=1}^{k}{C_i I_i} \geq k$. Taking Expectations we get that
\begin{eqnarray*}
k &\leq& {\sum_{i=1}^{k}\Exp{C_i \cdot I_i}} \\
&=& \sum_{i=1}^{k}{\prob{I_i = 0}\cdot\Exp{C_i\cdot 0 \mid  I_i = 0} + \prob{I_i = 1}\cdot\Exp{C_i \cdot 1\mid  I_i = 1}} \\
&=& \sum_{i=1}^{k} {p_i \Exp{C_i \mid I_i = 1}}.
\end{eqnarray*}
\qed\end{proof}

\begin{lemma}\label{lem:cond_Ci_bound} Given any $0<\gamma<\theoremTVal$, integer $n >  2\sqrt{e \ln{(1/\gamma) \cdot 1/\gamma}}$,  and  integer $0 \leq i \leq n/\theoremFVal$ we have that
$$\Exp{C_i \mid I_i = 1} \leq  1+\frac{\gamma n}{1 - \theoremFOverFVal\sqrt{\frac{1}{2}\ln{\frac{1}{1-\gamma}} + \frac{\ln{n}}{n}}}+ 1/n.$$

\end{lemma}
\begin{proof}
Define \begin{equation}a_i = \frac{i-1}{n} + \sqrt{\frac{1}{2}\ln{\frac{1}{1-\gamma}} + \frac{\ln{n}}{n}}.\label{eq:def_ai}
\end{equation} We have that
\begin{eqnarray} \Exp{C_i\mid I_i=1} &=& \quad \prob{(X_i < a_i)\mid I_i=1}\cdot\Exp{C_i \mid (X_i < a_i)\wedge I_i=1}\nonumber \\
&&\qquad + \prob{X_i \geq a_i\mid I_i=1}\cdot\Exp{C_i \mid (X_i \geq a_i)\wedge I_i=1}.\label{eq:scary_sum}\end{eqnarray}
For the second term, from Lemma \ref{lem:conditional_x_i_bound}, we
derive that $\prob{X_i \geq a_i \mid I_i = 1} \leq \frac{1}{n^2}$. In
addition, since the number of points to the right of (or at) the $i$th point from the left is $n-i+1$, we have that $$\Exp{C_i \mid  (X_i
\geq a_i)  \wedge  I_i = 1} \leq n-i+1 \leq n.$$ Hence, for the
second term
\begin{equation}\label{eq:first_term11}
\prob{X_i \geq a_i\mid I_i=1}\cdot\Exp{C_i \mid (X_i \geq a_i) \wedge I_i=1 } \leq \frac{1}{n}.
\end{equation}

For the first term we apply Lemmata \ref{lem:independent} (in Equation (\ref{eq:integral_1})) and \ref{lem:ci_cond_xi} (in Inequality (\ref{eq:integral_2})) to derive
\begin{eqnarray}
\lefteqn{\Exp{C_i \mid  (X_i < a_i) \wedge I_i = 1 }}  \nonumber \\
&=& \int_{0}^{a_i}{\prob{X_i = x \mid X_i < a_i \wedge I_i=1}\Exp{C_i \mid  X_i = x \wedge I_i=1} dx} \nonumber \\
&=& \int_{0}^{a_i}{\prob{X_i = x \mid X_i < a_i \wedge I_i=1}\Exp{C_i \mid  X_i = x} dx} \label{eq:integral_1} \\
&\leq& \int_{0}^{a_i}{\prob{X_i = x \mid X_i < a_i \wedge I_i=1}\left(1+\frac{\gamma}{1-t}(n-i)\right) dx} \label{eq:integral_2} \\
&\leq& \int_{0}^{a_i}{\prob{X_i = x \mid X_i < a_i \wedge I_i=1}\left(1+\frac{\gamma}{1-a_i}(n-i)\right) dx} \nonumber \\
&\leq& \left(1+\frac{\gamma}{1-a_i}(n-i)\right)\int_{0}^{a_i}{\prob{X_i = x \mid X_i < a_i \wedge I_i=1} dx} \nonumber \\
&=& 1+\frac{\gamma}{1-a_i}(n-i) \label{eq:c_i_cond_x_i_2}
\end{eqnarray}

By substituting (\ref{eq:first_term11}) and (\ref{eq:c_i_cond_x_i_2}) into Equation (\ref{eq:scary_sum}) we get that:
\begin{equation}
\Exp{C_i \mid I_i = 1} \leq 1+\frac{\gamma(n-i)}{1-a_i} ~+~  \frac{1}{n} .
\label{eq:c__i_cond12}
\end{equation}

Substituting the value of $a_i$ from Equation (\ref{eq:def_ai}) in Equation (\ref{eq:c__i_cond12}) we get
\begin{eqnarray}
\frac{\gamma(n-i)}{1-a_i}
&=& \frac{\gamma (n-i)}{1 - {\frac{i-1}{n}-\sqrt{\frac{1}{2}\ln{\frac{1}{1-\gamma}} + \frac{\ln{n}}{n}}}} \nonumber\\
&=& \frac{\gamma n(n-i)}{(n - i) + 1 -n\sqrt{\frac{1}{2}\ln{\frac{1}{1-\gamma}} + \frac{\ln{n}}{n}}} \nonumber\\
&\leq& \frac{\gamma n(n-i)}{(n - i) -n\sqrt{\frac{1}{2}\ln{\frac{1}{1-\gamma}} + \frac{\ln{n}}{n}}} \nonumber\\
&=& \frac{\gamma n}{1 - \frac{n\sqrt{\frac{1}{2}\ln{\frac{1}{1-\gamma}} + \frac{\ln{n}}{n}}}{n-i}} \nonumber\\
&=& \frac{\gamma n}{1 - \frac{n\sqrt{\frac{1}{2}\ln{\frac{1}{1-\gamma}} + \frac{\ln{n}}{n}}}{n-i}}. \label{eq:down}
\end{eqnarray}

Since  $i \leq n/\theoremFVal$ we get that:
\begin{eqnarray}
\frac{n\sqrt{\frac{1}{2}\ln{\frac{1}{1-\gamma}} + \frac{\ln{n}}{n}}}{n-i}
&\leq& \frac{n\sqrt{\frac{1}{2}\ln{\frac{1}{1-\gamma}} + \frac{\ln{n}}{n}}}{\theoremFOverOverFVal n } \nonumber\\
&=& \theoremFOverFVal\sqrt{\frac{1}{2}\ln{\frac{1}{1-\gamma}} + \frac{\ln{n}}{n}}. \label{eq:down_bound}
\end{eqnarray}

By the assumption that $\gamma<\theoremTVal$ and $n >  2\sqrt{e \ln{(1/\gamma) \cdot 1/\gamma}}$ it follows that Expression (\ref{eq:down_bound}) is strictly less than one.  Thus, we can substitute Expression (\ref{eq:down_bound}) into Equation (\ref{eq:down}) and derive
\begin{equation}
\label{eq:second_term11} 1+\frac{\gamma(n-i)}{1-a_i} \leq 1 + \frac{\gamma n}{1 - \theoremFOverFVal\sqrt{\frac{1}{2}\ln{\frac{1}{1-\gamma}} + \frac{\ln{n}}{n}}}.
\end{equation}

By substituting the upper bound in Equation (\ref{eq:second_term11}) into Inequality (\ref{eq:c__i_cond12}) we derive the statement of the Lemma.
\qed\end{proof}

\begin{lemma}\label{lem:conditional_x_i_bound}
For any $1 \leq i \leq n$,
$$Prob\left[X_i \geq \frac{i-1}{n} + \sqrt{\frac{1}{2}\ln{\frac{1}{1-\gamma}} + \frac{\ln{n}}{n}} ~ \middle\vert I_i = 1\right] \leq \frac{1}{n^2}$$
\end{lemma}
\begin{proof}
 Define $a_i = \frac{i-1}{n} + \sqrt{\frac{1}{2}\ln{\frac{1}{1-\gamma}} + \frac{\ln{n}}{n}}$. From Bayes rule and Lemmata \ref{lem:min_pi} and \ref{lem:x_i_bound} we get that
\begin{eqnarray}
\prob{X_i \geq a\bigr\rvert I_i = 1} &=& \frac{\prob{X_i \geq a_i \wedge I_i = 1}}{\prob{I_i = 1}} \nonumber \\
&\leq& \frac{\prob{X_i \geq a_i}}{\prob{I_i = 1}} \nonumber \\
&\leq& \frac{\frac{1}{n^2}(1-\gamma)^n}{(1-\gamma)^{n-1}} \label{eq:x_i_bound1} \\
&\leq& \frac{1}{n^2} \ .\nonumber
\end{eqnarray}
We use Lemma \ref{lem:x_i_bound} to bound the numerator in Inequality (\ref{eq:x_i_bound1}) and Lemma \ref{lem:min_pi} to bound the denominator in Inequality (\ref{eq:x_i_bound1}).
\qed\end{proof}

\begin{lemma}\label{lem:x_i_bound}
For $1\leq i \leq n$, $$\prob{X_i \geq \frac{i-1}{n} + \sqrt{\frac{1}{2}\ln{\frac{1}{1-\gamma}} + \frac{\ln{n}}{n}}~} \leq \frac{1}{n^2}(1-\gamma)^n$$
\end{lemma}
\begin{proof}
The probability that $X_i \ge p$ for some $p\in [0,1)$ is
exactly the probability that for some $k \le i-1$ exactly $k$ of the random points lie in $[0,p]$. That is
\begin{equation}
\prob{X_i \geq p} = \sum_{k=0}^{i-1}{{n\choose k}}p^i(1-p)^{n-k} \label{eq:binomial}
\end{equation}
Notice that the right hand side of Equation (\ref{eq:binomial}) is exactly the probability that a Binomial random variable with $n$ trials and success probability $p$ ($Bin(n,p)$) is at most $i-1$.
So we can apply Hoeffding's inequality
\footnote{Recall that Hoeffding's inequality for a Binomial random variable is $\prob{Bin(n, p) \leq (p-\epsilon)n} \le  e^{-2\epsilon^2 n}$. We use this with $\epsilon = p - \frac{i-1}{n}$.}
to get
\begin{eqnarray*}
\prob{X_i \geq p}
&=& \prob{Bin(n, p) \leq i-1} \\
&=& \prob{Bin(n, p) \leq \left(p - \left(p - \frac{i-1}{n}\right)\right)n} \\
&\leq & e^{-2\left(p - \frac{i-1}{n}\right)^2 n}
\end{eqnarray*}
By choosing $p = \frac{i-1}{n} + \sqrt{\frac{1}{2}\ln{\frac{1}{1-\gamma}} + \frac{\ln{n}}{n}}$ we get that
\begin{eqnarray*}
\prob{X_i \geq p} &\leq& e^{-2\left(p - \frac{i-1}{n}\right)^2 n} \\
&=& e^{-n\ln{\frac{1}{1-\gamma}} - 2\ln{n}} \\
&=& \frac{(1-\gamma)^n}{n^2}
\end{eqnarray*}
finishing the proof of the lemma.
\qed\end{proof}

\begin{lemma}\label{lem:independent}
For any $0 \leq i \leq n$, conditioned on the event   $X_i = x$,
the random variables $C_i$ and $I_i$ are independent. Formally,
$$\prob{C_i \leq k \mid (X_i =x) \wedge I_i} = \prob{C_i \leq k \mid (X_i =x)}.$$
\end{lemma}
\begin{proof}
Given that $X_i=x$, the set $A$ of the $i$ items arriving at times $\le x$, and the arrival times of the items in $A$ (that is $X_1,\ldots,X_{i-1}$), the action of the greedy algorithm on
the $i$th item arriving at $x$ is determined (that is $I_i$ is determined).
 On the other hand,
$C_i$ depends only on the arrival times of the points not in $A$ (that is $Z\setminus A$).
Since this holds for any set $A$, and for any arrival times of these points, it also holds  without conditioning on $A$.
\qed\end{proof}

\begin{lemma}\label{lem:ci_cond_xi}
For any $0 \leq i \leq n$,
$$E[C_i \mid X_i =x] \leq 1 + \frac{\gamma}{1-x}(n-i)$$
\end{lemma}
\begin{proof}
Conditioning on $X_i=x$, there are exactly $n-i$ points each distributed uniformly in $[x,1)$.
Let $\set{Z_1, \dots Z_{n-i} }$ be $n-i$ independent random variables each distributed uniformly in $[x,1)$. Since $C_i$ contains the $i$th point and any of the following points the falls in $[x, x+ \gamma)$ we have that
\begin{eqnarray*}
\Exp{C_i \mid X_i = x} &=& 1 + \Exp{\size{\set{Z_i < \gamma}}} \\
&=& 1 + \sum_{i = 1}^{n-i}{\prob{Z_i < \gamma}} \ .
\end{eqnarray*}
If $x \le 1-\gamma$ then $\prob{Z_i < \gamma} = \frac{\gamma}{1-x}$, otherwise $\prob{Z_i < \gamma} = 1$, therefore
\begin{eqnarray*}
\Exp{C_ i \mid  X_i = x} &=& 1 + \sum_{i = 1}^{n-i}\min\left({1, \frac{\gamma}{1-x}}\right) \\
&\leq& 1 + \sum_{i = 1}^{n-i}{\frac{\gamma}{1-x}} \\
&=& 1 + \frac{\gamma(n-i)}{1-x}
\end{eqnarray*}
\qed\end{proof}

\begin{lemma}\label{lem:iF_eq}
For all $i,k \geq 1$ such that $i+k \le n$,
$$\Exp{m(\set{X_{1}, \ldots, X_{k}}, \gamma)} = \Exp{m(\set{X_{i + 1}, \ldots, X_{i+k}}, \gamma)}.$$
\end{lemma}
\begin{proof}
The number of $\gamma$-independent points are invariant under translation and rotation, {\sl i.e.},
for any real $x$ and for every $\ell$, $1 \leq \ell\leq n$, and for any  set of points $\set{x_1,\ldots, x_\ell}$,
 \begin{equation} m(\set{x_1,\ldots, x_\ell},\gamma) =m(\set{x-x_1,\ldots, x-x_\ell},\gamma). \label{eq:minvariant} \end{equation}

Assume that $i+k = n$. The vectors $(X_1,\ldots, X_k)$ and  $(1-X_n,\ldots, 1-X_{n-k+1})$ have the same distribution, and the lemma follows from (\ref{eq:minvariant}).

Otherwise, ($i+k < n$), we condition on $X_{i+k+1} = x$. It now follows that the vectors $(X_1,\ldots, X_k)$ and  $(x-X_{i+k},\ldots, x-X_{i+1})$ have the same distribution. Therefore, it follows from (\ref{eq:minvariant}) that $$\Exp{m(\set{X_{1}, \ldots, X_{k}}, \gamma)\mid X_{i+k+1}=x} = \Exp{m(\set{X_{i + 1}, \ldots, X_{i+k}}, \gamma)\mid X_{i+k+1}=x}.$$
Since the above holds for every $x$, the lemma follows.
\qed\end{proof}

We now give several technical lemmata required to conclude the proofs above.
\begin{lemma}\label{len:1_min_1_n}
For any $0 < \alpha < 1$
$$(1-\alpha)^{1/\alpha} \geq\frac{1}{e} \cdot (1-\alpha)$$
\end{lemma}
\begin{proof}
For any $x$ such that  $1-\alpha \leq x \leq 1$ we have that $1 \leq 1/x \leq 1/(1-\alpha)$. So by integrating from $1-\alpha$ to $1$ we get that

\begin{align}
\; &&\int_{1-\alpha}^{1} \! \frac{1}{x}\, \mathrm{d}x &\leq  \int_{1-\alpha}^{1} \! 1/(1-\alpha) \, \mathrm{d}x \nonumber  \\
\Rightarrow && \ln(1) - \ln (1-\alpha) &\leq \alpha/(1-\alpha)  \nonumber \\
\Rightarrow  && \ln (1-\alpha) &\geq  -\frac{1}{1/\alpha-1}  \nonumber \\
\Rightarrow && e^{\ln (1-\alpha)} &\geq  e^{-\frac{1}{1/\alpha-1}}  \nonumber \\
\Rightarrow && (1-\alpha) &\geq  e^{-\frac{1}{1/\alpha-1}} \label{eq:exp03} \\
\Rightarrow &&(1-\alpha)^{1/\alpha-1} &\geq \frac{1}{e} \label{eq:exp04} \\
\Rightarrow &&(1-\alpha)^{1/\alpha} &\geq \frac{1}{e}\cdot (1-\alpha) \label{eq:exp05} \ .
\end{align}
Inequality (\ref{eq:exp04}) follows from (\ref{eq:exp03}) by taking the $(1/\alpha-1)$ power of both sides and (\ref{eq:exp05}) follows from (\ref{eq:exp04}) by multiplying both sides by $(1-\alpha)$.
\qed\end{proof}

\begin{lemma}\label{lem:arithmetic_xlnx}
For $x \geq 1/e$, $$\sqrt{e \cdot x\ln (1/x)} $$
is monotonically decreasing,
and for any $x > 0$, $$\sqrt{e \cdot x\ln(1/x)} \leq 1$$
\end{lemma}
\begin{proof}
Let $f(x)=x \ln (1/x)$. Then $f'(x) = -\ln{x} - 1$ which is positive
for any $0 < x < 1/e$ and negative for any $x > 1/e$. Therefore the maximum of $f$ is obtained at $x = 1/e$. It follows that for all $x>0$
$$
\sqrt{e \cdot x\ln(1/x)} \leq \sqrt{e \frac{1}{e} \ln (e)} = 1 \ .
$$
\qed\end{proof}

\begin{lemma}\label{lem:1_min_gamma}
For any $n \geq 2$ and any $0<\gamma \leq 0.5$
$$(1-\gamma) \geq \frac{1}{1+\gamma n} \ .$$
\end{lemma}
\begin{proof}
It suffices to prove the above inequality for $n=2$, for $n>2$ the right hand side can only decrease whereas the left hand side does not depend on $n$.

Now, $(1-\gamma)(1+2\gamma) = 1 +\gamma -2\gamma^2$, and $\gamma-2\gamma^2\geq 0$ for all $0<\gamma\leq 0.5$ so the lemma holds.
\qed\end{proof}

\begin{lemma}\label{lem:logngamman}
For any  $0< \gamma < 1/e$  we have that
$$\frac{n\ln{n}}{(n+1/\gamma)^2} \leq \left(\frac{\beta(\gamma)}{\theoremFVal}\right)\cdot\gamma\ln{(1/\gamma)} + \left(\frac{\beta(\gamma)\ln{(\beta(\gamma))}}{\theoremFVal}\right)\cdot\gamma,$$
where $\beta(\gamma)$ is defined to be $1 + \frac{2}{\ln{(1/\gamma)} - 1}$.
\end{lemma}
\begin{proof}

Fix $0<\gamma<1/e$, define $h(n) = \frac{n\ln{n}}{(n+1/\gamma)^2}$ and
let $n_0 = \argmax_{n\geq 1} h(n)$. Note that $\beta(\gamma)>1$ for $0<\gamma<1/e$, we show below that
\begin{equation} 1/\gamma \leq n_0 \leq \beta(\gamma)/\gamma \ . \label{eq:boundsn0}\end{equation}
Assume that Equation (\ref{eq:boundsn0}) holds. As both $n\ln{n}$ and ${(n+1/\gamma)}^{2}$ are monotonically increasing in $n$, it follows that for all $n\geq 1$
\begin{eqnarray}
  h(n) &\leq& h(n_0) \nonumber \\
&=& \frac{n_0 \ln(n_0)}{(n_0+1/\gamma)^2} \nonumber \\
&\leq& \frac{(\beta(\gamma)/\gamma) \ln(\beta(\gamma)/\gamma)}{(1/\gamma + 1/\gamma)^2} \label{eq:hn02}\\
&=& \left(\frac{\beta(\gamma)\gamma}{4}\right)\ln(\beta(\gamma)/\gamma) \nonumber \\
&=& \frac{\beta(\gamma)}{4}\cdot\gamma\ln(1/\gamma)+\frac{\beta(\gamma)\ln(\beta(\gamma))}{4}\cdot\gamma, \label{eq:hn03}
\end{eqnarray}

where Inequality (\ref{eq:hn02}) follows from the two bounds in Equation (\ref{eq:boundsn0}), Equation (\ref{eq:hn03}) is the statement of the Lemma.

It remains to prove the inequalities in (\ref{eq:boundsn0}). We show below that for all $1 \leq n < 1/\gamma$: $h'(n)>0$, and that for all $n > \beta(\gamma)/\gamma$: $h'(n)<0$, this proves Equation  (\ref{eq:boundsn0}).

 The derivative of $h$ (with respect to $n$) is
\begin{eqnarray*}
h'(n) &=& \frac{(\ln{n} + 1) (n+1/\gamma)^2 - 2(n+1/\gamma) n\ln{n}}{(n+1/\gamma)^4} \\
&=& \frac{(\ln{n} + 1) (n+1/\gamma) - 2n\ln{n}}{(n+1/\gamma)^3} \\
&=& \frac{(n + 1/\gamma) + (1/\gamma-n)\ln{n}}{(n+1/\gamma)^3}.
\end{eqnarray*}
 as the denominator of $h'(n)$ is positive, the nominator determines the sign. Ergo,
it is enough to look at the sign of
$$\frac{1}{\gamma} \cdot \bigr(\gamma n+1   + (1-\gamma n)\ln{n}\bigr),$$
and, again, as $1/\gamma >0$, this is equal to the sign of
$$k(n)=\gamma n+1   + (1-\gamma n)\ln{n}.$$

For $1 \leq n < 1/\gamma$, it must be that $k(n)>0$, and hence $h'(n) > 0$.

As $\beta(\gamma)>1$ and $n \geq \beta(\gamma)/\gamma$ we have that $n > 1/\gamma$ and $n\gamma \geq \beta(\gamma) > 1$. Now, we have that
\begin{eqnarray*}
k(n) &=& \gamma n+1   + (1-\gamma n)\ln(n)\\
&=&(\gamma n -1)\left(1 + \frac{2}{\gamma n - 1} - \ln{n}\right) \\
&<& (\gamma n -1)\left(1 + \frac{2}{\beta(\gamma) - 1} - \ln{(1/\gamma)}\right) \\
&=& \frac{\gamma n - 1}{\beta(\gamma) - 1}\bigr(2 - (\beta(\gamma) - 1)(\ln{(1/\gamma)}-1)\bigr) \\
&=& \frac{\gamma n - 1}{\beta(\gamma) - 1}\left(2 - \frac{2}{\ln{(1/\gamma)} - 1}\cdot (\ln{(1/\gamma) - 1})\right) \\
&=& 0.
\end{eqnarray*}

thus concluding the proof of (\ref{eq:boundsn0}) and the lemma.

\qed\end{proof}

\section{Expected size of the Maximum capacity \lowercase{$d$} $\gamma$-Independent Set (\lowercase{$d$} identical machines)} \label{sec:d_width_uppper}
\newcommand{\MacroDLowerVal}{\frac{e\sqrt{0.5}}{2\pi}}
\newcommand{\MacroMidLowerBound}{\frac{e}{2\pi}}
\newcommand{\MacroMidUpperBound}{\frac{\sqrt{2\pi}}{e^2}}
\newcommand{\BinProb}[3]{{{#1} \choose {#3}}{{#2}^{#3}(1-{#2})^{{#1}-{#3}}}}
\newcommand{\BinProbExplit}[4]{{{#1} \choose {#3}}{{#2}^{#3}(1-{#2})^{#4}}}

A capacity $d$ $\gamma$-independent set is a set of ``feasible rentals", given that $d$ items can be rented, and each item is rented for a period of length $\gamma$. Equivalently to the definition in Section \ref{sec:def}, a capacity $d$ $\gamma$-independent set is a set of points $S \subset [0,1)$,  such that given any subset of intervals $I$, $|I|>d$,  $I\subset \{[t, t+\gamma) \mid t\in S\}$, the intersection of all intervals in $I$ is empty.  The unit interval graph is defined by intervals of length $\gamma$ whose left endpoints are the points of $Z$.

Let $m_d(Z,\gamma)$ denote the size of
a maximum capacity $d$ $\gamma$-independent subset  of a set $Z$.
In this section we study the expectation of $m_d(Z,\gamma)$ when the points of $Z$ are chosen uniformly at random in $[0,1)$.
Specifically we prove:

\begin{theorem}\label{theorem:DWidthIndependetSet}
 Let  $Z$ be a set of $n$ points chosen uniformly at random in $[0,1)$, and let
 $\gamma = 1/k$, for some integer $k\geq 2$. Then we have
\begin{equation}\label{eq:dMachinesLower}
\Exp{m_d(Z,\gamma)}{Z} \geq
\min{(n,\, d\cdot1/\gamma)}\cdot\left(1-\Theta\left(\sqrt{\frac{\ln{d}}{d}}\right)\right).
\end{equation}
\end{theorem}

\begin{proof}
 Let $Y$ be a random subset of $Z$ of size $\min{(n, T)}$ where\\ $T= (d - \sqrt{3d\ln{d}})\cdot1/\gamma + 1$. Since
 for any subset $Y \subseteq Z$, $m_d(Y, \gamma) \leq m_d(Z, \gamma)$ a lower bound on $m_d(Y, \gamma)$ is also a lower bound on $m_d(Z, \gamma)$. As $Z$ is a random set of points and $Y$ is a random subset of $Z$, choosing $Z$ first and then choosing $Y$ gives the same distribution on $Y$ as does simply choosing $Y$ at random.

Let $Y' = \set{x\in Y \mid \left|[x-\gamma, x) \cap Y\right| < d}$ be the subset of $Y$ containing only points $x$ that do not have $d$ points in an interval of length $\gamma$ ending at $x$. So  $Y'$ is a capacity $d$ $\gamma$-independent set (as the greedy algorithm applied to $Y'$ will pick all the points).  Therefore for any fixed $Z$,
\begin{equation}\label{eq:YPrimeInd}
\Exp{m_d(Y', \gamma)}{Y\subseteq Z} = \Exp{\size{Y'}}{Y\subseteq Z}.
\end{equation}
Since $Y' \subseteq Y \subseteq Z$ we have
\begin{equation}\label{eq:parsedChain}
\Exp{\Exp{m_d(Y', \gamma)}{Y\subseteq Z}}{Z} \leq \Exp{\Exp{m_d(Y, \gamma)}{Y\subseteq Z}}{Z} \leq \Exp{m_d(Z, \gamma)}{Z}
\end{equation}
From (\ref{eq:YPrimeInd}) and (\ref{eq:parsedChain}) follows that it  suffices to show a lower bound on $\Exp{\Exp{\size{Y'}}{Y\subseteq Z}}{Z}$.
\begin{eqnarray}
\nonumber\lefteqn{\Exp{\Exp{\size{Y'}}{Y\subseteq Z}}{Z}} \\
&=& \sum_{y\in Y}{1\cdot\prob{\left|[y-\gamma, y) \cap Y\right| < d}}\nonumber
\\\nonumber &\geq& \size{Y}\cdot{\prob{Bin(\size{Y}-1,\gamma) \leq d-1}}
\\\nonumber &=& \size{Y}\cdot\left(1 - {\prob{Bin(\min{(n,T)-1},\gamma) \geq d}}\right)
\\\nonumber &\geq& \size{Y}\cdot\left(1 - {\prob{Bin(T-1),\gamma) \geq d}}\right)
\\\nonumber &=& \size{Y}\cdot\left(1 - {\prob{Bin(T-1,\gamma) \geq \left(1 + \frac{d-(T-1)\gamma}{(T-1)\gamma}\right)(T-1)\gamma}}\right)\
\\&=& \size{Y}\cdot\left(1 - {\prob{Bin(T-1,\gamma) \geq \left(1 + \frac{\sqrt{3d\ln{d}}}{(T-1)\gamma}\right)(T-1)\gamma}}\right)\ . \label{eq:lowerBeforeChernoff}
\end{eqnarray}
The derivation following the 4th line above follows since $d=(1+(d-(T-1)\gamma)/(T-1)\gamma)(T-1)\gamma$.
Applying the multiplicative form of the Chernoff bound  ($\prob{X\geq (1+\epsilon)\mu}\leq e^{-\frac{\epsilon^2\mu}{3}}$ when $X$ is sum of $n$ IID random variables and $\mu = \Exp{X}$
 on  $Bin(T-1,\gamma)$  with $\epsilon = \frac{\sqrt{3d\ln{d}}}{(T-1)\gamma}$ we obtain

\begin{eqnarray*}
\Exp{\size{Y'}}{Y\subseteq Z}
    &\geq& \size{Y}\cdot\left(1 - e^{-\frac{1}{3}\cdot\left(\frac{\sqrt{3d\ln{d}}}{(T-1)\gamma}\right)^2\cdot(T-1)\gamma}\right)
    \\&=& \size{Y}\cdot\left(1 - e^{-\frac{d\ln{d}}{(T-1)\gamma}}\right)
    \\&=& \size{Y}\cdot\left(1 - e^{-\frac{d\ln{d}}{d - \sqrt{3d\ln{d}}}}\right)
    \\&\geq& \size{Y}\cdot\left(1 - e^{-\frac{d\ln{d}}{d}}\right)
    \\&=& \size{Y}\cdot\left(1 - 1/d\right) \ .
\end{eqnarray*}
Substituting $\size{Y} = \min{\Bigr(n, (d - \sqrt{3d\ln{d}})\cdot1/\gamma + 1\Bigr)}$ we get
\begin{eqnarray*}
\Exp{\size{Y'}}{Y\subseteq Z}
&\geq& \min{\Bigr(n, \left(d - \sqrt{3d\ln{d}}\right)\cdot1/\gamma\Bigr)}\cdot\left(1 - 1/d\right)
\\&\geq& \min{\Bigr(n, d/\gamma\Bigr)}\cdot \left(1 - \frac{\sqrt{3\ln{d}}}{\sqrt{d}}\right)\cdot(1 - 1/d)
\\&\geq& \min{\Bigr(n, d/\gamma\Bigr)}\cdot \left(1 - \frac{\sqrt{3\ln{d}}}{\sqrt{d}}- 1/d\right)
\\&=& \min{\Bigr(n, d/\gamma\Bigr)}\cdot \left(1 - \Theta\left(\sqrt{\frac{\ln{d}}{d}}\right)\right) \ .
\end{eqnarray*}
\qed\end{proof}

\begin{theorem}\label{theorem:DWidthIndependetSetUpper}
Let  $Z$ be a set of $n=d\cdot 1/\gamma$ points chosen uniformly at random in $[0,1)$ and let $\gamma = 1/k$ (for some integer $k\geq 2$). Then we have
\begin{equation}\label{eq:dMachinesUpper}
\Exp{m_d(Z,\gamma)}{Z}  \leq  n\cdot\left(1-\Theta\left(\frac{1}{\sqrt{d}}\right)\right).
\end{equation}
\end{theorem}
\begin{proof}
For any $0 \leq i < 1/\gamma$ define $Z_i \subseteq Z$ to be the set $\set{z \in Z \mid z\in[i\gamma ~, ~(i+1)\gamma)}$, all the points arriving at the $i$-th slice of size $\gamma$.

Obviously $\bigcup_{1 \leq i < 1/\gamma}{Z_i} = Z$.
Therefore
\begin{equation}
m_d(Z,\gamma) \leq \sum_{i=0}^{1/\gamma}{m_d(Z_i,\gamma)}
\label{eq:lowerSum}\end{equation}
To give un upper bound on the size of $m_d(Z_i,\gamma)$ notice that if $\size{Z_i} \leq d$ then $Z_i$ is a capacity $d$ $\gamma$-independent set, thus $m_d(Z_i,\gamma) = \size{Z_i}$.
Otherwise, any subset $B$ of $Z_i$ is capacity $d$ $\gamma$-independent set iff $\size{B} \leq d$.
Thus $m_d(Z_i,\gamma) = d$. Combining these observations we get
\begin{equation}
\Exp{m_d(Z_i,\gamma)}{Z} = \Exp{min(\size{Z_i}, d)}{Z}\ .
\label{eq:mdMin}\end{equation}
$\size{Z_i}$ count the number of points that fall into the line segment $[i\gamma,(i+1)\gamma)$ when we throw $n=d/\gamma$ points uniformly at random into the line segment $[0,1)$. Ergo,  $|Z_i|$ is distributed as $Bin\left(d\cdot1/\gamma,\, \gamma\right)$, therefore
$$\Exp{m_d(Z_i,\gamma)}{Z} = \Exp{\min(Bin\left(d\cdot1/\gamma,\, \gamma\right), d)}.$$
From Lemma (\ref{lem:binWithoutMid}) (see below) and the equation above we conclude that $\Exp{m_d(Z_i,\gamma)}{Z} \leq d\bigr(1 - \MacroMidUpperBound\cdot \frac{\sqrt{1-\gamma}}{\sqrt{d}}\bigr)$. Substituting this bound into (\ref{eq:mdMin}) and then into (\ref{eq:lowerSum}) we obtain
\begin{eqnarray*}
 \Exp{m_d(Z,\gamma)}{Z} \leq \sum_{i=0}^{1/\gamma-1}{\Exp{m_d(Z_i,\gamma)}{Z}}
 &\leq& 1/\gamma \cdot d\left(1 - \MacroMidUpperBound\cdot\frac{\sqrt{1-\gamma}}{\sqrt{d}}\right)
 \\ &\leq& n\left(1 - \MacroMidUpperBound\cdot\frac{1}{\sqrt{d}}\right)
 \\ &=& n(1-\Theta(\sqrt{1/d})),
\end{eqnarray*}
finishing the proof of the theorem.
\qed\end{proof}

The following lemma deals with the following experiment: Toss $n$ coins with probability $p$ of heads, return the number of heads if the number of heads $\leq \mu=np$ (the expectation), otherwise return $\mu$. What is the expected value of this experiment relative to $\mu$?
The probability that the outcome exceeds  $\mu + t \sigma$ decreases exponentially with $t$, where $\sigma$ is the standard deviation ($\sqrt{\mu}$ in our case). This suggests that the expected difference is $O(1)$ standard deviations, which is the claim of the next Lemma:

\begin{lemma}\label{lem:binWithoutMid}
For any $n \in \mathrm{N}$, $0<p\leq 1$, and $\mu=np$, we have that
$$\mu\left(1-\frac{e}{2\pi}\frac{\sqrt{1-p}}{\sqrt{\mu}}\right)
\leq \Exp{\min{\left(Bin(n,p), \mu\right)}}
\leq \mu\left(1-\frac{\sqrt{2\pi}}{e^2}\frac{\sqrt{1-p}}{\sqrt{\mu}}\right)\, .  $$
\end{lemma}
\begin{proof}
\begin{eqnarray}
\lefteqn{\Exp{\min{\left(Bin(n,p), \mu\right)}}} \nonumber\\
&=& \sum_{k=0}^{n}{\min{(\mu, k)\cdot\BinProb{n}{p}{k}}}\nonumber\\
&=& \sum_{k=0}^{\mu}{k\BinProb{n}{p}{k}}
    ~+~ \mu\cdot\prob{Bin(n,p)\geq \mu +1}\nonumber\\
&=& \mu\sum_{k=1}^{\mu}{\frac{k}{\mu}\cdot\frac{n}{k}\BinProbExplit{n-1}{p}{k-1}{n-k}}p
    ~+~ \mu\cdot\prob{Bin(n,p)\geq \mu +1}\nonumber\\
&=& \mu\sum_{k=1}^{\mu-1}{\frac{np}{\mu}\cdot\BinProb{n-1}{p}{k}}
    ~+~ \mu\cdot\prob{Bin(n,p)\geq \mu +1}\nonumber\\
&=& \mu\cdot\prob{Bin(n-1,p)\leq \mu -1}
    ~+~ \mu\cdot\prob{Bin(n,p)\geq \mu +1}\label{eq:midBinTarget}
\end{eqnarray}

We describe a single experiment with two disjoint events $A$, and $B$, such that $\prob{A}=\prob{Bin(n-1,p)\leq \mu -1}$, and $\prob{B}=\prob{Bin(n,p)\geq \mu +1}$. This will be useful as then the probability of ($A$ or $B$) is simply the sum of these probabilities. The experiment
is to toss $n$ coins, where the probability of heads is $p$, \begin{itemize} \item Event $A$ occurs if amongst the first $n-1$ results there were no more than $\mu - 1$ heads, $\prob{A}=\prob{Bin(n-1,p)\leq \mu -1}$. \item Event $B$ occurs if in total, over $n$ coin tosses, there were at least $\mu + 1$ heads, $\prob{B}=\prob{Bin(n,p)\geq \mu +1}$. \item It is easy to see that if $A$ holds then $B$ cannot occur. Likewise, if $B$ holds then $A$ cannot occur. \end{itemize}

Thus,
\begin{eqnarray}\label{eq:midBinRightSide}
\lefteqn{\prob{A \vee B} = \prob{A} + \prob{B}} \nonumber \\ &=& \prob{Bin(n-1,p)\leq \mu -1} + \prob{Bin(n,p)\geq \mu +1}
\end{eqnarray}
On the other hand, the complimentary event to $A\vee B$, $\overline{A \vee B}$, occurs if during the first $n-1$ tosses there were exactly $\mu$ heads and the last coin toss gave tails. Therefore
\begin{eqnarray}
\prob{A \vee B} &=& 1 - \prob{~\overline{A \vee B}~} \nonumber
\\ &=& 1 - \prob{Bin(1, p) = 0} \cdot\prob{Bin(n-1, p) = \mu}\nonumber
\\ &=& 1- (1-p)\cdot\prob{Bin(n-1, p) = \mu}\label{eq:midBinLeftSide}
\end{eqnarray}
Substituting (\ref{eq:midBinLeftSide}) into (\ref{eq:midBinRightSide}) and the result into (\ref{eq:midBinTarget}) we obtain
\begin{equation}\label{eq:midBinLastSub}
\Exp{\min{\left(Bin(n,p), \mu\right)}} = \mu\big(1 - (1-p)\cdot\prob{Bin(n-1, p) = \mu}\big).
\end{equation}
To finish the proof we bound $(1-p)\cdot\prob{Bin(n-1, p) = \mu}$, recall that $\mu=np$:
\begin{eqnarray}
\lefteqn{(1-p)\cdot\prob{Bin(n-1, p) = \mu}} \nonumber
\\&=& (1-p)\BinProb{n-1}{p}{np} \nonumber
\\&=& \frac{(n-1)!}{(np)!\cdot(n-1-np)!}p^{np}(1-p)^{n-np}\nonumber
\\&=& \frac{n-np}{n}\cdot\frac{n!}{(np)!\cdot(n-np)!}\cdot p^{np}(1-p)^{n-np}\nonumber
\\&\geq& \frac{n-np}{n}\cdot\frac{\sqrt{2\pi}\cdot n^n\sqrt{n}}{e^n}~\nonumber
    \\ && \qquad\cdot\,\frac{e^{np}}{e\cdot(np)^{np} \cdot \sqrt{np}}
        \,\cdot\,\frac{e^{n-np}}{e\cdot (n-np)^{n-np} \cdot\sqrt{n-np}}\cdot p^{np}(1-p)^{n-np} \label{eq:stirlingLeft}
\\ &=& \frac{\sqrt{2\pi}}{e^2}
    \cdot\frac{n-np}{n}\cdot \frac{\sqrt{n}}{\sqrt{np} \cdot \sqrt{n-np}}
    \cdot\frac{n^n}{n^{np}\cdot p^{np}\cdot n^{n-np} \cdot(1-p)^{n-np}} \cdot p^{np}(1-p)^{n-np}\nonumber
\\ &=& \MacroMidUpperBound
    \cdot\frac{n-np}{n}\cdot \frac{\sqrt{n}}{\sqrt{np} \cdot \sqrt{n-np}}
    \nonumber
\\ &=& \MacroMidUpperBound
    \cdot \frac{\sqrt{1-p}}{\sqrt{\mu}}\label{eq:binMidFinalLower}
\end{eqnarray}
where Inequality (\ref{eq:stirlingLeft}) follows from Stirling's lower bound on $n!$ in the numerator and Stirling's upper bound on $(np)!$ and $(n-np)!$ in the denumerator.

A similar proof, using Stirling's upper bound on $n!$ and Stirling's lower bound on $(np)!$ and $(n-np)!$ yields
\begin{equation}\label{eq:binMidFinalUpper}
(1-p)\cdot\prob{Bin(n-1, p) = \mu} \leq \MacroMidLowerBound\frac{\sqrt{1-p}}{\sqrt{\mu}}
\end{equation}
Assigning bound (\ref{eq:binMidFinalUpper}) and (\ref{eq:binMidFinalLower}) into equation (\ref{eq:midBinLastSub}) gives the statement of the lemma.
\qed\end{proof}

\end{document}